\documentclass[journal]{IEEEtran}

\ifCLASSINFOpdf
\else
   \usepackage[dvips]{graphicx}
\fi
\usepackage{url}

\usepackage{graphicx, tikz}
\usetikzlibrary{arrows.meta}
\usepackage{amssymb}
\usepackage{xcolor}
\usepackage{mathtools}
\usepackage{hyperref}
\usepackage{csquotes}
\usepackage{amsmath}
\usepackage{amsthm}
\usepackage{utfsym, adforn}
\usepackage{wrapfig}
\usepackage[T1]{fontenc}
\usepackage[utf8]{inputenc}
\usepackage{algorithm}
\usepackage{algpseudocode}
\usepackage{graphicx}
\usepackage{subcaption}
\usepackage{booktabs, stfloats, tabularx}
\usepackage{soul, cancel}

\DeclareMathOperator*{\sinc}{sinc}

\newcommand{\x}{\mathbf{x}}
\newcommand{\y}{\mathbf{y}}

\renewcommand{\Return}{\textbf{Return }}

\newtheorem{proposition}{Proposition}[section]

\theoremstyle{remark}
\newtheorem{remark}{Remark}[section]

\begin{document}
\pagestyle{empty}

\title{Autoregressive Stochastic Clock Jitter Compensation in Analog-to-Digital Converters}

\author{Daniele Gerosa, Rui Hou, Vimar Björk, Ulf Gustavsson, and Thomas Eriksson
\thanks{D. G. is a post-doctoral researcher at the Communication, Antennas and Optical Networks, Electrical Engineering dept., Chalmers University of Technology, Göteborg, Sweden (e-mail: daniele.gerosa@chalmers.se).}
\thanks{R. H. is a RF Technology Expert at Ericsson AB, Stockholm, Sweden.}
\thanks{V. B. is a Radio Architecture Expert at Ericsson AB, Göteborg, Sweden.}
\thanks{U. G. is a Senior Specialist at Ericsson AB, Göteborg, Sweden as well as Guest Researcher at the Communication, Antennas and Optical Networks, Electrical Engineering dept., Chalmers University of Technology, Göteborg, Sweden.}
\thanks{T. E. is full professor at the Communication, Antennas and Optical Networks, Electrical Engineering dept., Chalmers University of Technology, Göteborg, Sweden (e-mail: thomase@chalmers.se).}}

\markboth{}
{Shell \MakeLowercase{\textit{et al.}}: Bare Demo of IEEEtran.cls for IEEE Journals}
\maketitle
\thispagestyle{empty}

\begin{abstract}
This paper addresses the mathematical modeling and compensation of stochastic discrete-time clock jitter in analog-to-digital converters (ADCs). We model the stochastic clock jitter as a first-order autoregressive (AR(1)) process, and we propose two novel, computationally efficient, pilot-assisted dejittering algorithms for baseband signals: one based on solving a sequence of weighted least-squares problems, and another that exploits the correlated jitter structure via a Kalman filter-based routine. We also propose a conditional maximum-likelihood estimator for the autoregressive parameters, enabling near-optimal Kalman-filter performance even when such parameters vary over time. We further provide a mathematical analysis of the induced linearization errors, and we complement the theory with synthetic simulations to evaluate the proposed techniques across different scenarios.

The proposed techniques are shown to yield a \(1-15\)~dB improvement in signal-to-noise-and-distortion ratio (SINADR) and \(0.02-1.6\)~dB in symbol error vector magnitude (EVM), depending on impairment severity and pilot density. The Kalman smoother generally provides superior performance by leveraging additional temporal information.

\end{abstract}


\IEEEpeerreviewmaketitle
\section{Introduction}
Analog-to-digital converters are hardware devices that convert analog signals into digital signals. They are never ideal, and they typically introduce various imperfections into the sampled signals. These imperfections arise from non-ideal hardware and may appear as direct-current (DC) voltage offset bias, quantization noise, aliasing, nonlinear distortion due to clipping, and irregularities in the sampling instants (jitter) \cite{Valkama_Springer_Hueber}. This paper focuses on the latter imperfection, which arises from random phase fluctuations in the sampling clock generated by the clock source and its conditioning/distribution circuitry.

From a theoretical perspective, sampling instants that deviate from the ideal multiples of the sampling interval do not necessarily preclude the reconstruction of bandlimited analog signals, provided that the average sampling rate is large enough. This is established by various extensions of the classical interpolation results by Shannon, Nyquist, Whittaker and Kotelnikov \cite{Marvasti}. The caveat is that the timing deviations must be known. In practical applications, these timing deviations are not available a priori and must therefore be inferred from measurements. 
Considerable effort has been devoted to the modeling, estimation and compensation of sampling-clock jitter in ADCs. Several works use out-of-band pilot (or training) tones injection, as in \cite{Syrjala_Valkama}\cite{Rutten_Breems_vanVeldhoven}\cite{Towfic_Ting_Sayed}\cite{GerosaAnttilaEriksson}, where the baseband payload is typically filtered to isolate the pilot tone. The jitter realizations are estimated from this tone and then used to compensate (i.e., dejitter) the payload. A time-domain technique that time-multiplexes pilot samples to unknown data samples is described in \cite{Sung_Choi} in the context of time-interleaved (ti) ADCs. Other statistical signal processing methods proposed in the literature aim to directly estimate the payload without first necessarily estimating the jitter realizations. Minimum Mean Squared Error (MMSE) estimators are proposed in \cite{Weller_Goyal}\cite{Nordio_Chiasserini_Viterbo}\cite{Testoni_Speciale_Ridolfi_Pouzat}\cite{IbarraManzanoetal} under different payload models.

In parallel with algorithmic methods, the literature has also devoted substantial attention to how sampling clock jitter should be modeled statistically. While static (i.e., constant in time) jitter is mostly relevant in the field of ti ADCs \cite{Divi_Wornell}, stochastic jitter is often modeled as white noise \cite{Weller_Goyal}\cite{Nordio_Chiasserini_Viterbo}\cite{Testoni_Speciale_Ridolfi_Pouzat}\cite{IbarraManzanoetal}\cite{Maetal}\cite{Araghi_Akhaee_Amini} or as colored process with a pronounced \(1 / f^2 \) spectral decay component \cite{Towfic_Ting_Sayed}, depending on whether the jitter is caused by an internal or an external clock reference \cite{Zanchi_Samori}. Clocks with external references can also add an additional low-offset \(1/f^{3}\) region in their spectrum, which is often neglected for simplicity (e.g., by retaining only the \(1/f^{2}\) component) \cite{Zanchi_Samori}\cite{DaDalt_Harteneck_Sandner_Wiesbauer}. These spectral features are commonly observed in oscillators \cite{Zanchi_Samori}. 

To the best of our knowledge, autoregressive (AR) processes were first employed in \cite{Chang_Lin_Wang_Lee_Shih} to model sampling jitter in sigma-delta modulators. Wiener-process models are also common in phase-noise studies \cite{Gävert_Coldrey_Eriksson}\cite{Petrovic_Rave_Fettweis}. In this work, we model sampling-clock jitter as a first-order autoregressive process (AR(1)), which is mathematically tractable and can reproduce the relevant spectral characteristics with an appropriate choice of parameters. More details are provided in Section~\ref{sec:sys_model}.

We exclude effects such as nonlinearities, phase noise, and imperfect bandlimiting. Such effects will also degrade the signal quality. However, most of them are relatively orthogonal to the jitter, and can be treated separately using independent techniques\footnote{As an example, this is the case in many Ericsson products.}. In the present work, we assume that the residual distortion from these other impairments can be modeled as additive and incorporated into the thermal noise term (see e.g. \cite{Moghaddametal}).

The main theoretical contributions of this paper are summarized as follows:
\begin{itemize}[]
\item[\adfhalfrightarrowhead] Following \cite{Chang_Lin_Wang_Lee_Shih}, we model the stochastic clock jitter as an autoregressive process of order \(1\) with model parameter close to \(1\).
\item[\adfhalfrightarrowhead] We propose and compare two novel jitter tracking and compensation algorithms based on pilot samples. The first uses Kalman filtering and smoothing to exploit the jitter statistical structure. The second uses an optimally weighted least-squares approach and is more ``model-agnostic'' (Section \ref{sec:main_sec}). Depending on the scenario, both techniques are shown to yield a \(1-15\) dB improvement in Signal-to-Noise-and-Distortion Ratio (SINADR) and \(0.02-1.6\)~dB in symbol error vector magnitude (EVM) over the uncompensated case, assuming sufficiently high symbol or sample pilot density.
\item[\adfhalfrightarrowhead] We propose a Maximum Likelihood Estimation technique for autoregressive parameter estimation within the Kalman smoother routine.
\end{itemize}
Furthermore, the analysis also includes:
\begin{itemize}[]
\item[\adfhalfrightarrowhead] rigorous justifications for the approximations arising in Taylor expansions and an outline of the regime in which these approximations are valid (Section \ref{subsec:approx_err});
\item[\adfhalfrightarrowhead] synthetic simulations to evaluate the techniques presented (Section \ref{sec:num_sim}) and to compare them with existing methods.
\end{itemize}
\subsection{Notation and symbols}
ADCs act on continuous-time signals by discretizing and digitizing them, so that the output of an ADC is a discrete sequence of values \( \{x(t_n)\}_{n \in \mathbb{N}} \) when the input is the analog signal \( x(t) \). If \(x(t) \) is a stochastic process, we can model the ADC output as a sequence of random variables. We will denote sequences (finite or infinite) by bold lowercase letters \( \mathbf{x} \coloneqq (x_1, x_2, \dots ) \). If \( \x\) and \( \y \) are two such sequences, \( \x \odot \y \) indicates their Hadamard componentwise product. Bold uppercase letters \( \mathbf{A}, \mathbf{B}, \dots \) will represent matrices; the symbol ``\( \approx \)'' stands for ``approximately equal to'' while the symbol ``\( \ll \)'' means ``much less than''. We will use \( \widehat{\cdot} \) for estimates of quantities, and \( \widetilde{\cdot} \) for noisy measurements of quantities.

If \( X  \) is a complex-valued random variable defined on the probability space \( ( \Omega, \Sigma, P) \), then \( \mathbb{E}[X] \) denotes its expected value, \( \mathbb{E}[X] = \int_\Omega \mathfrak{Re}[X] \, d P + i \int_\Omega \mathfrak{Im}[X] \, d P \) and \( \text{var}(X) \) denotes its variance, \( \sigma_X ^2 = \text{var}(X) = \mathbb{E}[|X - \mathbb{E}[X]|^2]   \). 

For a set \(E\), \( \chi_E \) indicates its characteristic function and \( |E| \) its cardinality.

\subsubsection*{Derivatives}
By \(x' _n = x' (n T_s)  \) we mean the first-order time derivative of the realization of the process \( x(t) \) evaluated at times \( n T_s \); for a bandlimited Gaussian process this operation is well-defined \cite{Belyaev} and is equivalent to the \(n\)-th sample of the discrete derivative operator \(D\) applied to the (infinite) sequence of samples, i.e. \[ \frac{d}{d t} x(t) \big|_{t = n T_s} = x' (n T_s)= (D \mathbf{x})_n. \] This derivative operator \( D\) acts as time domain convolution \[ (D \x)_n = \frac{1}{T_s} \sum_k h_k x_{n-k}  \] with the ideal non-causal bandlimited derivative filter given by \( h_k = (-1)^k / k \) for \( k \ne 0 \) and \( h_0 = 0 \); its frequency response is \( H_D (\Omega) = i \Omega \). These expressions are commonly found and used in the relevant literature, cf. \cite{Salib_Flanagan_Cardiff}\cite{Divi_Wornell}\cite{Oppenheim_Schafer}.

Oftentimes throughout the text, we will encounter continuous-time processes \( y \) (Wiener, Ornstein-Uhlenbeck, white noise) whose derivative may not exist in the standard pointwise sense. However, upon bandlimiting and sampling, their discrete \( D \y \) counterparts are well-defined. To keep the notation as light as possible, we will write, with a slight abuse of notation, \( y'_n \) in lieu of \( (D \y)_n \).

\section{Signal and system model}
\label{sec:sys_model}
We consider baseband analog signals that are modeled as continuous-time bandlimited complex-valued stationary Gaussian processes \( x(t, \omega) \) so that, for each time instant \( t \in T \subseteq \mathbb{R} \), the random variable \( (\omega \mapsto x(t, \omega)) \) is Gaussian with \(0\) mean and variance \(\sigma_x ^2\). We also assume that \(x\) has flat power spectral density \( \mathcal{S}_x (f) = \sigma_x ^2 / (2W) \chi_{ \{|f| \le W \}} (f)  \). This standard model, commonly used to describe signals in communication applications (cf. \cite{Proakis}), possesses analytically well-behaved sample paths: by Theorem 11 in \cite{Belyaev} the maps \( t \mapsto x(\omega, t) \) are indeed holomorphic for almost every \( \omega \in \Omega \). Therefore the ``n-th derivatives objects'' \( x^{(n)} \) are well-defined for all \(n \in \mathbb{N} \).

\subsection{Stochastic jitter description and problem statement}
An ideal ADC samples signals at uniformly spaced instants. This means that \(t_{n+1} - t_n \) is constant for all \(n\ge 0 \). In real-world applications, however, imperfections in the local oscillator (LO) circuitry introduce jitter, causing the actual sampling instants to deviate from the ideal in a non-uniform way: i.e., the actual sampling takes places at times \( \tilde{t}_n = n T_s + \xi_n \) where \( T_s \) is the ADC sampling interval and \( \xi_n \) is the time-dependent (stochastic) jitter. The literature presents different models for the (discrete) process \( \{ \xi_n \}_{n \ge 1} \): in \cite{Nordio_Chiasserini_Viterbo} \( \xi_n  \sim \mathcal{U}([-a,a]) \) and independent while in \cite{Weller_Goyal}\cite{Testoni_Speciale_Ridolfi_Pouzat}\cite{Araghi_Akhaee_Amini} the \( \xi_n   \) are modeled as i.i.d. Gaussian with \(0\) mean and variance \( \sigma^2 \). \cite{Towfic_Ting_Sayed} describes \( \xi_n \) as a ``slowly varying'' Gaussian process. To the best of our knowledge, only \cite{Chang_Lin_Wang_Lee_Shih} models the time domain clock jitter as an autoregressive process, while the Wiener process is more broadly used e.g. in \cite{Petrovic_Rave_Fettweis}\cite{Gävert_Coldrey_Eriksson}\cite{Khanzadi_Krishnan_Kuylenstierna_Eriksson}. Several sources characterize the jitter process in the frequency domain by deriving formulas for its power spectral density. These expressions turn out to coincide with or be very close to the power spectral densities of AR(1) processes; compare, for example, \cite[eq.9]{Chang_Lin_Wang_Lee_Shih} with \eqref{AR1discpsd}, and \cite[eq.4]{DaDalt_Harteneck_Sandner_Wiesbauer} with \eqref{eq:AR1contpsd}.

Therefore, in the present work we model the jitter process as a discrete autoregressive process of order \(1\) AR(1) with evolution relation given by\begin{equation} \label{eq:AR_ev} \xi_n = \varphi \xi_{n-1} + \epsilon_n, \end{equation} where \( \varphi \approx 1 \) yet \( \varphi < 1\), \( \epsilon_n \sim \mathcal{N}(0, \sigma_\epsilon ^2) \) and \( \xi_0 \sim \mathcal{N}(0, \sigma_\epsilon ^2 / (1 - \varphi^2)) \) (stationary initialization). It can be shown \cite{Kay} that the power spectral density of this process is \begin{equation} \label{AR1discpsd} \mathcal{S}_\xi (e^{i \omega}) =   \frac{\sigma_\epsilon ^2}{1 + \varphi^2 - 2 \varphi \cos(\omega)}, \ \omega \in [-\pi,\pi],  \end{equation} from which it is straightforward to observe that much of the process power is concentrated near \( \omega = 0 \) (or \(0\) Hz): indeed \begin{equation} \label{eq:AR1approxpsd}  \mathcal{S}_\xi (e^{i \omega})  = \frac{ \sigma_\epsilon ^2}{ (1 - \varphi)^2 + 2 \varphi (1 - \cos(\omega)) }  \approx \frac{ \sigma_\epsilon ^2}{ \omega^2 } \end{equation} where the last approximation in \eqref{eq:AR1approxpsd} holds as long as \( \varphi \approx 1 \), \( \omega \approx 0 \) but \( |\omega| > |1-\varphi|\). Moreover notice that \(\mathcal{S}_\xi (e^{i 0}) = \sigma_\epsilon ^2 / (1 - \varphi)^2 \). These observations show that the discrete process \eqref{eq:AR_ev} can accurately model the jitter spectral decay as experimentally measured \cite{Zanchi_Samori}.
In addition, we remark that the Ornstein-Uhlenbeck process, which can be seen as the continuous counterpart of \eqref{eq:AR_ev}, has the Lorentzian profile \begin{equation} \label{eq:AR1contpsd}
  \mathcal{S}_{\text{OU}} (\omega) =  \frac{\sigma^2}{\gamma^2 + \omega^2 }
\end{equation} as spectral density, from which it is again visible the \(1/\omega^2\) decay, now for all \(\omega\) ``large enough''.

If the jitter is not too large, with respect to the sampling interval, its effect on the discretized signal is often \cite{Towfic_Ting_Sayed}\cite{Testoni_Speciale_Ridolfi_Pouzat}\cite{Nordio_Chiasserini_Viterbo}\cite{elbornsson} approximated via a first-order Taylor expansion: \begin{equation} \label{jitter_TaylorExp}  x(n T_s + \xi_n)  \approx x(n T_s ) + \xi_n x' (n T_s). \end{equation}
As we will show in Proposition \ref{sec_order_varEst}, the higher order terms in the Taylor expansion can be considered in some sense negligible. 

From here on, we will also set \begin{equation} \label{jitter_TaylorExp_and_noise} y_n \coloneqq x(n T_s ) + \xi_n x' (n T_s) + w_n  \end{equation}as our model under investigation, with \( w_n \) a sequence of independent zero-mean Gaussian random variables with variance \( \sigma_w^2 \) modeling additive white noise in measurements. It has been observed \cite{Towfic_Ting_Sayed} that \( \xi_n\) may be weakly correlated with the \(w_n\), but in this work, we will consider the two processes as independent.

The problem we address in this work is to estimate the process \( \xi_n \) given the measurements \( y_n \ \forall n \) together with pilot samples \( x_n \) for \( n \in I \subseteq \mathbb{N} \) and, as a byproduct, to derive a dejittered estimation of \( x_n \).

\subsection{How much jitter is ``too much'' jitter?}

One way to quantify how much jitter can be tolerated by the system is by setting a target SINADR level. Since the disturbance is modeled as additive (cf. \eqref{jitter_TaylorExp}), we can write the SINADR for our signal and jitter model (excluding quantization noise) \begin{equation} \label{SINADR} \begin{split} \text{SINADR}_{\text{dB}} & = 10 \log_{10} \left[ \frac{\text{var}(x_n)}{\text{var}(w_n) + \text{var}(\xi_n x' _n) } \right] \\ & = 10 \log_{10} \left[ \frac{\text{var}(x_n)}{\text{var}(w_n) + \text{var}(\xi_n) \text{var}( x' _n) } \right] \\ &  = 10 \log_{10} \left[ \frac{3(1 - \varphi^2) \sigma_x ^2 }{3(1 - \varphi^2) \sigma_w ^2 + 4 \pi^2 W^2  \sigma_\epsilon ^2 \sigma_x ^2} \right],  \end{split}  \end{equation} where we used \( \sigma_\xi ^2 = \sigma_\epsilon ^2 / (1 - \varphi^2) \), \( \mathbb{E}[\xi_n] = \mathbb{E}[x' _n] = 0 \) and the fact that \( \xi_n \) and \(x_n '\) are independent. By solving eq. \eqref{SINADR} for \( \sigma_\epsilon ^2 \) it is easy to see how much jitter could be tolerated (for a desired SINADR level). At the same time, for the first-order Taylor expansion \eqref{jitter_TaylorExp} to hold, the jitter \( \xi_n \) cannot be too large; therefore one reasonable assumption is to make the ``small jitter'' hypothesis (SJH), in line with the literature surveyed so far, \[\label{small_jitter_hyp} \sigma_\xi / T_s \ll 1.  \tag{SJH}\] We will refer to the left-hand side of the latter as \emph{jitter percentage}. Notice that when \( \sigma_w = 0 \), \eqref{SINADR} reduces to \[ -20 \log_{10} \left[ 2 \pi  W \sigma_\xi  / \sqrt{3}   \right] \approx -20 \log_{10} \left[ 2 \pi  W \sigma_\xi   \right] + 4.77 \text{ dB}, \] which is the standard jitter SNR formula found in classical literature (see e.g. \cite{Shinagawa_Akazawa_Wakimoto}), applied to bandlimited signals.

The interaction between the two additive distortions in \eqref{jitter_TaylorExp_and_noise} will also play a role in our analysis. For convenience we introduce a Noise-to-Distortion Ratio (NDR), defined as the ratio between the measurement white noise power and the jitter-induced distortion power. Although NDR is not a standard performance metric in the ADC literature, it provides a concise, dimensionless parameter that captures the relative strength of the two additive impairments in our model and helps organize the numerical results. In particular, we set
\begin{equation} \begin{split}
\mathrm{NDR}_{\mathrm{dB}} & \coloneqq 10\log_{10}\!\left(\frac{\mathrm{var}(w_n)}{\mathrm{var}(\xi_n x'_n)}\right) \\ &
= 10 \log_{10} \left[ \frac{3 \sigma_w ^2 (1 - \varphi^2) }{4 \pi^2 W^2 \sigma_x ^2 \sigma_\epsilon^2}  \right]. \end{split}
\end{equation}
 Since the observation model \eqref{jitter_TaylorExp_and_noise} contains two additive impairments, their relative strength affects the behavior of the proposed estimators. For this reason, we will report and interpret our experimental results also as a function of \(\mathrm{NDR}_{\mathrm{dB}}\). Note that when \( \text{NDR}_{\text{dB}} = 0 \), the two additive terms in \eqref{jitter_TaylorExp_and_noise} have the same variance.

\subsection{Approximation errors analysis} \label{subsec:approx_err}
There are two main approximations used in our analysis, whose errors need to be rigorously analyzed. In the first-order expansion in \eqref{jitter_TaylorExp} under \eqref{small_jitter_hyp} the neglected higher-order terms are, in some sense, negligible, as the following Proposition clarifies: 

\begin{proposition} \label{sec_order_varEst} Let \( \x = \{x_n \}_{n \in \mathbb{N}} \) be a discrete zero-mean, wide-sense stationary (WSS) bandlimited Gaussian process, and \( \boldsymbol{\xi} = \{  \xi_n \}_{n \in \mathbb{N}} \) a discrete autoregressive process of order \(1\) with variance \( \sigma_{\xi_n}^2 \equiv \sigma_{\xi}^2 \), independent of \(\x\). Assume that \eqref{small_jitter_hyp} holds; then we have \begin{equation} \label{small_jitt}
\emph{var}( \xi_n ^2 x'' _n) \ll \emph{var}( \xi_n  x' _n) \quad \forall \, n. 
\end{equation}
\end{proposition}

The second source of error comes from the derivative of \(y\) in \eqref{jitter_TaylorExp_and_noise}, that will replace the unknown \( x'_n \) in our algorithms: 

\begin{equation}
y_n ' = x_n ' + u'_n + w'_n,
\end{equation}
 where we set \(\mathbf{u} = \boldsymbol{\xi} \odot D \x \).  The term \( w'_n \) is violet noise bandlimited to \( 1 / 2 T_s \), whose variance equals \( \pi^2 \sigma_w ^2 / 3 T_s ^2  \). For a fixed NDR level \(r\) in dB, the white noise variance satisfies \( \sigma_w ^2 = 10^{r/10} \text{var}(\xi x')  \). Therefore \begin{equation} \label{violet_noise_varEst}  \sigma_{w'} ^2 = \frac{\pi^2 10^{r/10}\sigma_\xi ^2  }{3 T_s ^2} \sigma_{x'}^2  \approx 3.29 \frac{\sigma_\xi ^2 }{T_s ^2} 10^{r/10} \sigma_{x'}^2     \end{equation}and thus \( \sigma_{w'} ^2 \ll \sigma_{x'}^2 \) provided that \eqref{small_jitter_hyp} holds and the white noise is not too strong.
 The term \( D (\boldsymbol{\xi} \odot \mathbf{x}')_n \) has a more convoluted structure, but its impact can be explicitly assessed, as shown in Proposition \ref{yprim_approx_xprim}, the proof of which can be found in the appendix. The power relation between \(u'_n\) and \( x'_n \) can be evaluated explicitly by means of equation \eqref{var_derEst}. Figure \ref{fig:rat_cond} illustrates this ratio for a few values of \( \varphi \), \( \sigma_\xi / T_s \) and \( WT_s \).

\begin{proposition} \label{yprim_approx_xprim}
    Let \( \x = \{x_n\}_{n \in \mathbb{N}}\) be a discrete zero-mean, wide-sense stationary bandlimited Gaussian process (sampled every \(T_s\) seconds), with derivative \( \x' = \{x_n '\}_{n \in \mathbb{N}}\), and \( \boldsymbol{\xi} = \{\xi_n\}_{n \in \mathbb{N}} \) a discrete autoregressive process of order \(1\) with parameter \( \varphi \in (0, 1) \) and variance \( \sigma_{\xi_n}^2 \equiv \sigma_{\xi}^2 \), independent of \( \x\). Let \( D\) denote the ideal time-domain derivative operator. Then 
\begin{equation}
    \label{var_derEst} \emph{var}( D (\boldsymbol{\xi} \odot \x ')_n) = \frac{\sigma_{\xi}^2}{T_s ^2} \frac{3 (1 - \varphi^2)}{32 \pi^4 (W T_s)^3 } I \sigma_{x'}^2
\end{equation}
where \[ I = \int_{-\pi}^{\pi} \int_{-\pi}^{\pi} \frac{ \omega^2 ((\omega - \nu)_{2 \pi})^2}{1 - 2 \varphi \cos(\nu) + \varphi^2} \chi_{\{|(\omega - \nu)_{2 \pi}| \le 2 \pi W T_s\} } \, d \omega \, d \nu \] is a Poisson kernel-type double integral.  
\end{proposition}

\begin{remark}
A stochastic version of the Peano remainder theorem for Taylor expansions holds, cf. Lemma 3 in \cite{Yang_Zhou_Wang}. Further results comparing (conditional) variances of the objects \( x(nT_s + \xi_n, \omega) \) and \( x(nT_s, \omega) + x'(nT_s, \omega) \xi_n \) would require some care. First of all, the measurability/randomness of the maps \( \omega \mapsto x(nT_s + \xi_n(\omega), \omega) \) or \( (\omega_1, \omega_2) \mapsto x(nT_s + \xi_n(\omega_2), \omega_1)  \) must be carefully assessed, depending on whether the jitter process is assumed to live in the same probability space of the process \(x\) or not. A standard sufficient condition for the measurability of \( x(t, \omega) \) on the whole product space \( \mathbb{R} \times \Omega \) is the left- or right-continuity in \(t\) to hold \emph{for all} \( \omega \in \Omega \) (see Lemma 6.4.6 in \cite{Bogachev}), which might not be the case for our signal model. In that case, one should replace \(x\) with an indistinguishable copy obtained by redefining it on a \(P\)-null set so that \(t\mapsto x(t,\omega)\) is left-/right-continuous for every \(\omega\).
\end{remark}

\begin{figure}[ht]
\centering
    \includegraphics[width=0.24\textwidth]{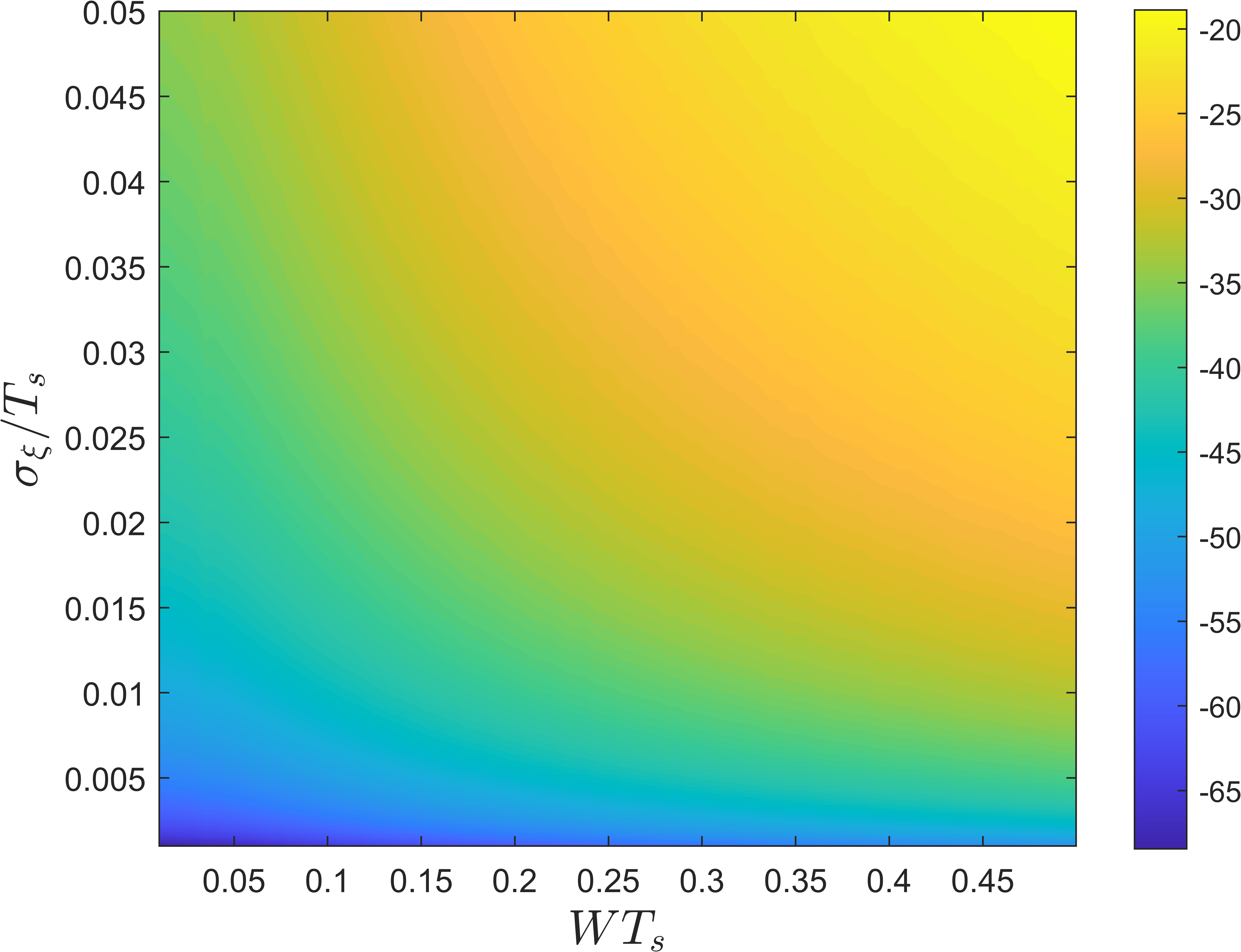}
    \includegraphics[width=0.24\textwidth]{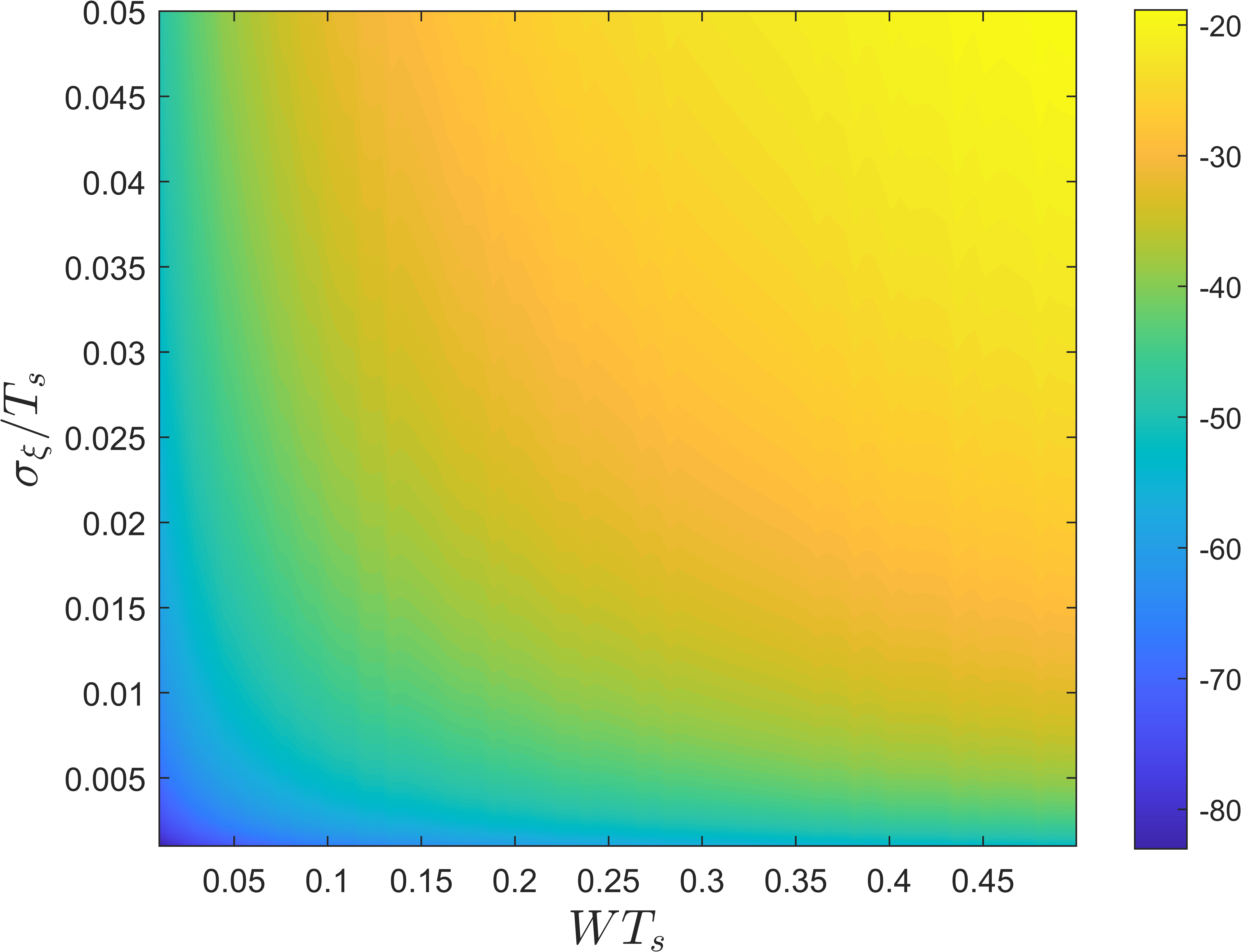}
    \caption{Level curves of the term multiplying \( \sigma_{x'}^2 \) in \eqref{var_derEst}; \( \varphi=0.95 \) left, \( \varphi=0.999 \) right, in dB. The double integral in \eqref{var_derEst} was numerically evaluated via nested \texttt{integral} in \texttt{MATLAB} and its accuracy was double-checked via a separate Monte Carlo where instances of the processes \( \mathbf{x} \), \( \boldsymbol{\xi} \) etc. were generated and the left-hand side of \eqref{var_derEst} estimated.}
    \label{fig:rat_cond}
\end{figure}

\section{Pilot-sample-based jitter compensation} 
\label{sec:main_sec}
This section contains the main contributions of this paper. We present two jitter tracking techniques able to mitigate the negative impact of this disturbance. Both algorithms assume knowledge of pilot samples at known time index positions, meaning that the true sample values at \( n \in I \subsetneq \mathbb{N}  \) will be assumed to be known a priori. The output stream of samples will be jittered, including the samples at pilot positions. A noisy estimate of the jitter process sample values can be obtained by rearranging equations \eqref{jitter_TaylorExp} and \eqref{jitter_TaylorExp_and_noise} at pilot positions: \begin{equation}  
\label{jitt_approx}\xi_n \approx \widetilde{\xi}_n = (y_n - x_n) / x' _n, \quad n \in I.
\end{equation}

Once the jitter process is estimated using one of the methods detailed in this section, the jitter distortion in \eqref{jitter_TaylorExp_and_noise} must be subtracted, the pilots removed, and the signal interpolated using e.g. the Gerchberg-Papoulis algorithm \cite{Papoulis_alg}\cite{Gerchberg} to ``fill the gaps'' left after the removal of pilot samples. Depending on the density of sample pilots, the ADC sampling frequency must be high enough to ensure the uniqueness of the bandlimited reconstruction. The workflow is summarized by the scheme in Figure \ref{fig:threebox}; the focus of this paper is within the green box.

\begin{figure}[ht]
  \centering
\begin{tikzpicture}[
  node style/.style = {
    draw, 
    minimum width=2.2cm, minimum height=2.2cm, text width=2.1cm, font=\small,
    align=center
  },
  arr/.style  = {
    thick, 
    -{Stealth[length=5pt,width=7pt]}
  }
]

% Place the three boxes at (0,0), (5,0), and (10,0)
\node[node style, draw=red, line width=1.2pt]    (A) at (0,0)   {Introduce known analog stimulus at transmitter side.};
\node[node style, draw=green, line width=1.2pt]  (B) at (3cm,0) {Use pilot samples to estimate \( \xi_n \, \forall \, n \)  and compensate the distorted signal.};
\node[node style, draw=red, line width=1.2pt] (C) at (6cm,0) {Remove sample pilots (if unwanted) and interpolate (e.g. via Gerchberg-Papoulis).};

% Arrows between them
\draw[arr] (A.east) -- (B.west);
\draw[arr] (B.east) -- (C.west);
\end{tikzpicture}
  \caption{Schematic three‐step process.}
  \label{fig:threebox}
\end{figure}

In contrast to pilot tone-based techniques \cite{Towfic_Ting_Sayed}, no knowledge of the signal's derivative is assumed by our techniques, i.e. \( x' _n \) is \emph{not} assumed to be known a priori at pilot positions or elsewhere.

\begin{figure}[!h]
    \centering
    \includegraphics[width=0.45\textwidth]{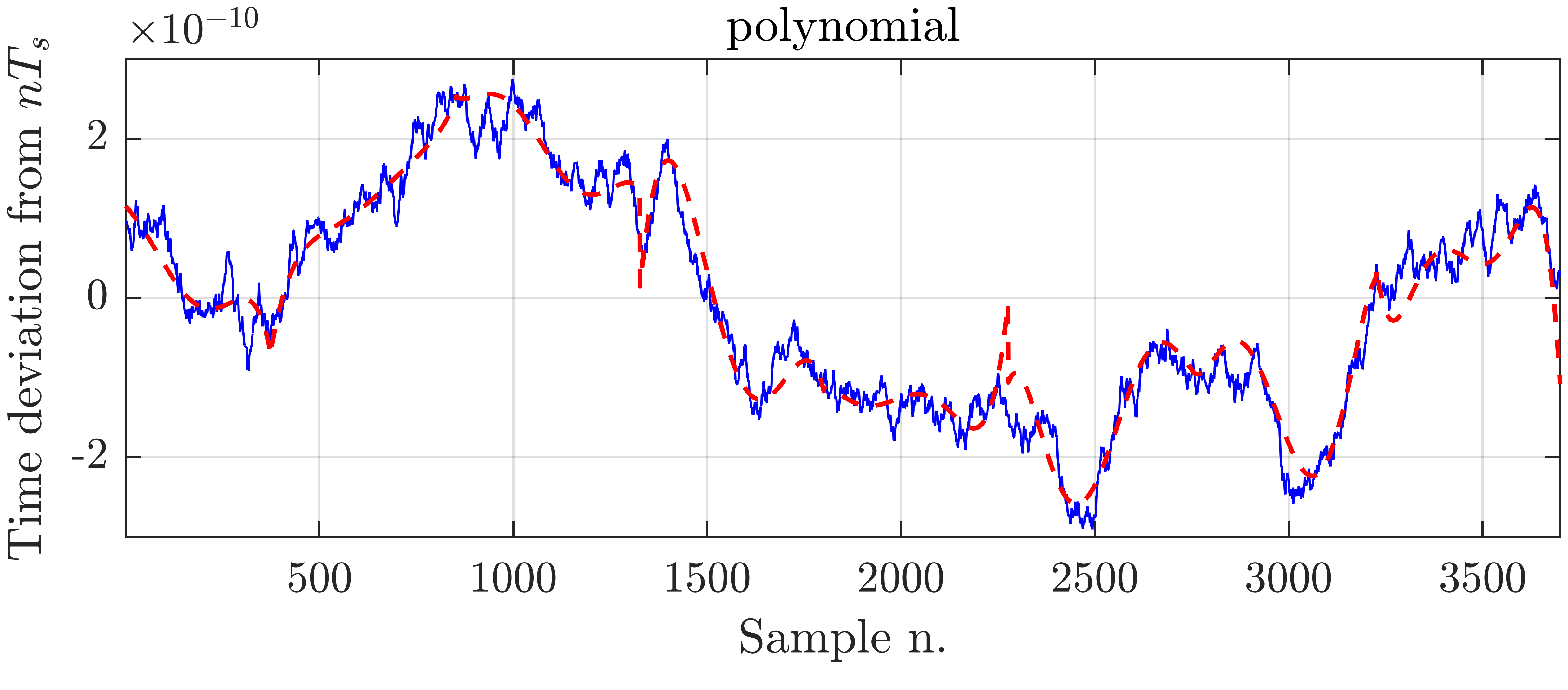}
    \includegraphics[width=0.45\textwidth]{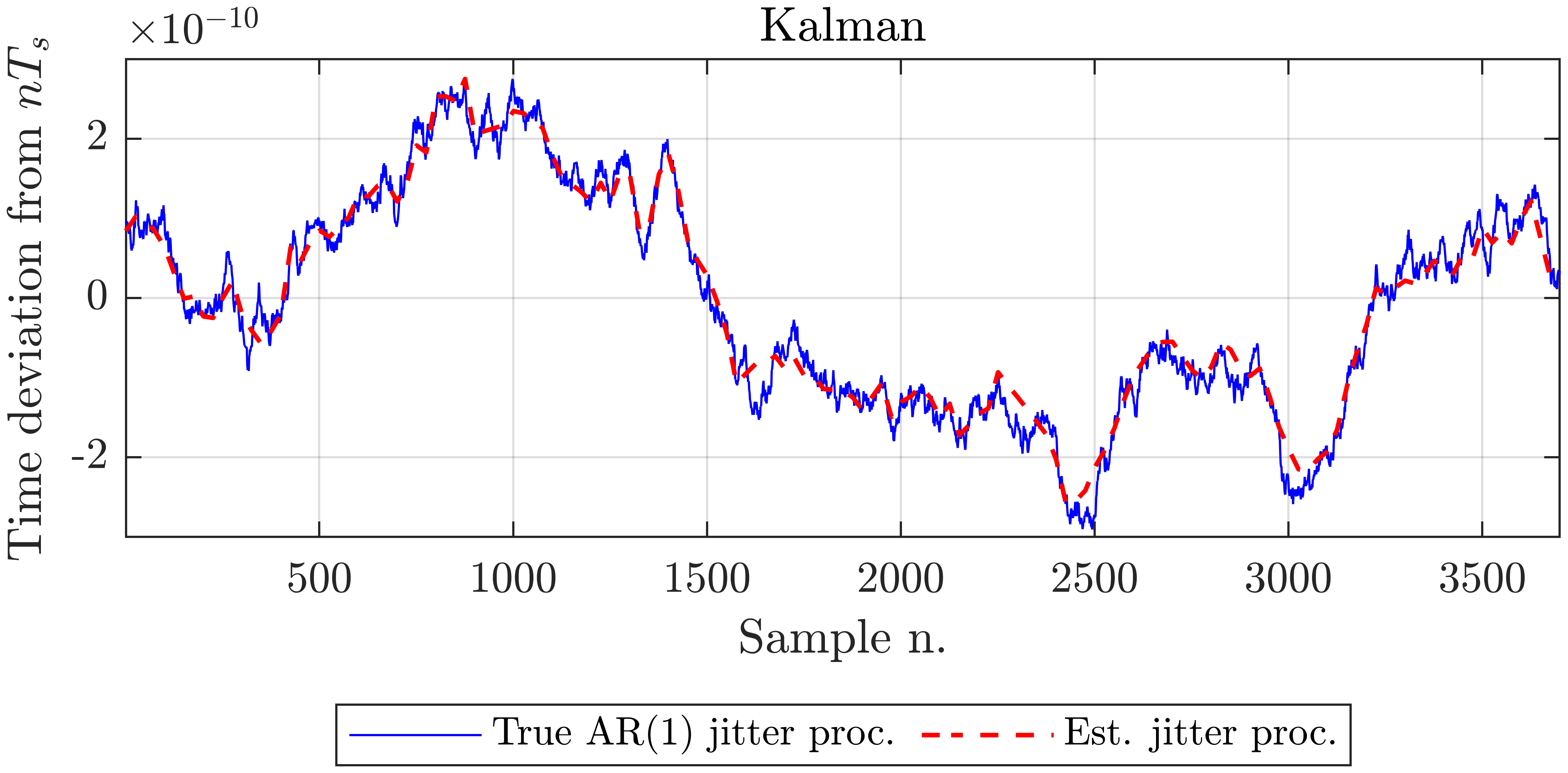}
    \caption{Jitter trackers comparison example, Algorithm \ref{alg:poly_alg} (left) vs Algorithm \ref{alg:KalmanF} (right).}
    \label{fig:example_polyvsKalman}
\end{figure}

 Moreover, the knowledge of sparse signal's sample values \( x_n \) at pilot positions is insufficient to compute its derivative at those positions. In our algorithms, we will use \( y' _n \) as a proxy instead, as justified by the analysis in Section \ref{subsec:approx_err}. For illustration, Figure \ref{fig:example_polyvsKalman} shows a tracking example: despite the large amount of missing data, both methods recover the overall behavior of the underlying autoregressive process.

\subsection{Robust blockwise polynomial interpolation}
\label{poly_tech}
In this section, the true latent jitter process is treated as blockwise deterministic, yet unknown, and slowly varying. Under the model \eqref{jitter_TaylorExp}, noisy estimates of the jitter at pilot positions can be obtained from \eqref{jitt_approx} as \( \widetilde{\xi}_n = \xi_n + w_n / x_n ' \). The additive noise is conditionally Gaussian but heteroscedastic because its variance is time-dependent; moreover, the unknown values \(x' _n \) are to be replaced with their estimations \( y'_n \). We set \( v_n = w_n / y'_n \). The idea is to perform an optimal short-time (i.e., blockwise) polynomial interpolation of the jitter process at the observation points \(\widetilde{\xi}_n \), \( n \in I \) for a given \(I\).

Let \( N \) be the total number of samples observed. Since we want the pilot samples to be uniformly spaced, we fix a value \( K \in \mathbb{N} \) (the number of data samples between pilot samples) and set \[ I = \{ n (K+1) + 1 \, : \, n = 0, \dots, \lfloor (N - 1)/ (K +1) \rfloor \};\] we can write \( I =  \bigcup_{l = 1}^L I_l \), where each \( I_l \) is a (sorted) set of strictly increasing integer indices such that the last element of \(I_l\) coincides with the first element of \(I_{l+1}\). Moreover, we require that each \( I_l\) contains the same number \(C\) of elements, except possibly for \(I_L\). We indicate indices belonging to \( I_l \) with \( i_c ^{(l)} \in I_l \). As noted earlier in this section, we assume that there exist \(L\) vectors \( \overline{\boldsymbol{\beta}}^{(l)} \) of polynomial coefficients such that \( \mathbf{M}\overline{\boldsymbol{\beta}}^{(l)} = \boldsymbol{\xi}^{(l)}  \), where \( \mathbf{M} \) is the rectangular measurements matrix \begin{equation} \label{meas_matr}
\mathbf{M} = \begin{bmatrix} 1 & i_1 ^{(1)} T_s & \dots & (i_1 ^{(1)} T_s)^d \\ 
1 & i_2 ^{(1)} T_s & \dots & (i_2 ^{(1)} T_s)^d \\ \vdots & \vdots & \ddots & \vdots \\
1 & i_C ^{(1)} T_s & \dots & (i_C ^{(1)} T_s)^d
\end{bmatrix}
\end{equation}
and \( \boldsymbol{\xi}^{(l)} =  (\xi_{i_1 ^{(l)}}, \dots, \xi_{i_C ^{(l)}})^t \). Our cost function is the conditional mean-square error 

\begin{equation} \label{poly_cond_exp}
 \mathbb{E}[\| \mathbf{M} \mathbf{S} \widetilde{\boldsymbol{\xi}}^{(l)} - \boldsymbol{\xi}^{(l)} \|_2 ^2 \, | \, \y '  ],
\end{equation}
that we seek to minimize as a function of \(\mathbf{S}\). We do not use ordinary least-squares because the measurement noise \(\widetilde{\boldsymbol{\xi}}^{(l)}\) has time-varying variance. Equation \eqref{poly_cond_exp} can be written as standard bias-variance decomposition (by expanding the square) \begin{equation} \label{bias+var} \| (\mathbf{I} - \mathbf{M} \mathbf{S} ) \boldsymbol{\xi}^{(l)} \|_2^2 + \mathbb{E}[\|\mathbf{M} \mathbf{S} \mathbf{v}^{(l)}\|_2 ^2 \, | \, \y ' ], \end{equation} and by imposing the (necessary) unbiasedness condition \( \mathbf{S}\mathbf{M}  = \mathbf{I}  \) (cf. (3.28) in \cite{Rao_Toutenberg}) we find that the second term in \eqref{bias+var} is minimized if \begin{equation} \label{weighted_MPinv}
\widehat{\mathbf{S}}_l = (\mathbf{M}^t \mathbf{W}_l  \mathbf{M})^{-1} \mathbf{M}^t \mathbf{W}_l  
\end{equation}
where \( \mathbb{E}[\mathbf{v}^{(l)}(\mathbf{v}^{(l)})^t ]^{-1} =  \mathbf{W}_l \propto \operatorname{diag}\bigl\{ (y'_{i^{(l)}_c})^2\bigr\}_{c=1}^C \), while the bias term vanishes. Notice that \( \widehat{\mathbf{S}}_l \mathbf{M}  = \mathbf{I} \); in addition we recall that \eqref{weighted_MPinv} solves the weighted least squares problem for heteroscedastic noise \cite{Aitken}.

The sought polynomial coefficients are thus computed via \( \widehat{\boldsymbol{\beta}}^{(l)} = \widehat{\mathbf{S}}_{l} \widetilde{\boldsymbol{\xi}}^{(l)}  \), leading to an explicit formula for the interpolating / denoising polynomial functions:
\begin{equation*}
p_l (t) \coloneqq \widehat{\beta}^{(l)} _1 + \widehat{\beta}^{(l)} _2 t + \dots + \widehat{\beta}^{(l)} _{d + 1} t^d.
\end{equation*}
The estimated jitter process at all time instants will hence be 
\begin{equation} \label{interp_values}
    \widehat{\xi}_n  = \begin{cases} p_1(nT_s), & 1 \le n \le i_C^{(1)} \\  p_{l} ((n - i_1^{(l)} + 1)T_s), & i_1 ^{(l)} < n \le  i_C ^{(l)}, \ l = 2, \dots, L. \end{cases}
\end{equation}
The dejittered signal is then estimated as: \begin{equation} \label{de_jitt_sig}
\widehat{x}_n = y_n - \widehat{\xi}_n y'_n.
\end{equation}
We summarize the algorithm in the following flowchart: 

\begin{algorithm}[ht]
\caption{Robust polynomial interpolation}
\label{alg:poly_alg}
\begin{algorithmic}[1]
\Require Observations \(y_n\), \(y'_n\), pilot indices \(I\) and \(x_n\) for \(n \in I \); number of samples \(N\), measurements matrix \( \mathbf{M} \) \eqref{meas_matr}.
\Statex
\For{\(l=1\) to \(L\)}
\State Form measurements vector \( \widetilde{\boldsymbol{\xi}}^{(l)} \) according to \eqref{jitt_approx}
\State Minimize \eqref{poly_cond_exp} in \( \mathbf{S} \): \( \widehat{\mathbf{S}}_l \gets (\mathbf{M}^t \mathbf{W}_l  \mathbf{M})^{-1} \mathbf{M}^t \mathbf{W}_l  \)
\State Interpolating polynomial coefficients \( \widehat{\boldsymbol{\beta}}^{(l)} \gets \widehat{\mathbf{S}_l}\widetilde{\boldsymbol{\xi}}^{(l)}  \)
\State Interpolated values at non pilots pos. \( \widehat{\boldsymbol{\xi}}^{(l)}  \) \eqref{interp_values}
\EndFor 
\Statex
\State \Return de-jittered signal \(\widehat{x}_n = y_n - \widehat{\xi}_n y'_n\) \eqref{de_jitt_sig} \( \forall \, n\).
\end{algorithmic}
\end{algorithm}

\subsection{Kalman smoother}

In this section, we apply Kalman filtering and smoothing to our jitter tracking problem, which are standard techniques that provide optimal recursive estimates for systems described using the state-space model formalism. The state-space model for the linearized jitter tracking problem under investigation is \begin{equation} \label{jitter_state-model}  \begin{cases} \xi_n = \varphi \xi_{n-1} + \epsilon_n \\ y _n = x_n + \xi_n x' _n + w_n,  
\end{cases} \end{equation} from which it should be obvious to the reader how Kalman filtering and smoothing can estimate the process \( \xi_n \) provided that sufficiently accurate values of the parameter \( \varphi \), \( \sigma_\epsilon ^2 \) and white noise variance \( \sigma_w ^2 \) are available.

Our implementation of the Kalman filter and smoother routine is described in Algorithm \ref{alg:KalmanF}. It follows the standard textbook routine (see, e.g., \cite{Simon}), but it includes a slight modification to handle missing data at non-pilot positions: the measurement covariances \( R_n \) are set to \( = +\infty \), making the Kalman gain null. This correctly encodes the fact that the autoregressive state is observed only at the pilot indices (i.e., measurement updates are performed only on pilots). Equivalently, between two consecutive pilot updates separated by \(K\) data samples, the AR(1) dynamics imply the \(K\)-step relation \begin{equation} \label{eq:subsampl_AR1} \xi_n = \varphi^{K+1}\xi_{n-K-1} + \nu_n, \quad \nu_n \sim \mathcal{N}\!\left(0,\ \sigma_\epsilon^2 \frac{1 - \varphi^{2(K+1)}}{1-\varphi^2}\right),\end{equation}
and the smoother interpolates the intermediate states using \eqref{jitter_state-model}. Notice that we replace the observation ``matrices'' \(x'_n\) with their estimates \(y'_n\). Algorithm~\ref{alg:KalmanF} clarifies our implementation.

\begin{algorithm}[ht]
\caption{(Inexact) Kalman Filter and Smoother flow for \eqref{jitter_state-model}}
\label{alg:KalmanF}
\begin{algorithmic}[1]
\Require Observations \(y_n\), \(y'_n\), pilot indices \(I\) and \(x_n\) for \(n \in I \); number of samples \(N\); AR(1) parameter \(\varphi\); AR process variance \( Q_n = \sigma_\epsilon ^2 \); measurement noise variance \(R_n = \sigma_w^2\) for \(n \in I \), \( R_n = +\infty \) elsewhere; initial state \(\widehat{\xi}_{0|0}\) and covariance \(P_{0|0}\).
\Statex
\State \textbf{Forward Pass}
\For{\(n=1\) to \(N\)}
    \If{\(n>1\)}
        \State \textbf{Prediction:}
        \State \(\widehat{\xi}_{n|n-1} \gets \varphi \widehat{\xi}_{n-1|n-1}\)
        \State \(P_{n|n-1} \gets \varphi^2 P_{n-1|n-1} + Q_n\)
    \Else
        \State Set \(\widehat{\xi}_{n|n-1} \gets \widehat{\xi}_{0|0}\)
        \State Set \(P_{n|n-1} \gets P_{0|0}\)
    \EndIf
    \State \textbf{Measurement Update:}
    \State Define \(H_n \gets y_n '\)
    \State Innovation: \(\widetilde{y}_n \gets y_n - x_n - H_n \widehat{\xi}_{n|n-1}\)
    \State Innovation covariance: \(S_n \gets |H_n|^2 P_{n|n-1} + R_n \)
    \State Kalman gain: \(K_n \gets P_{n|n-1} \overline{H_n} / S_n\)
    \State Update state estimate: \(\widehat{\xi}_{n|n} \gets \widehat{\xi}_{n|n-1} + K_n \widetilde{y}_n\)
    \State Update error covariance: \(P_{n|n} \gets P_{n|n-1} - K_nH_n P_{n|n-1}\)
\EndFor 

\Statex

\State \textbf{Backward Pass (Kalman Smoother)}
\State Let \(\widehat{\xi}_{N|N}\) be the final filter estimate from above
\State \(\widehat{\xi}_{N|N}^\text{smooth} \gets \widehat{\xi}_{N|N}\)

\For{\(n = N-1\) down to \(1\)}
    \State Compute smoother gain: \( C_n \gets \varphi P_{n|n} / P_{n+1|n}\)
    \State Smoothed state: \(\widehat{\xi}_{n|N}^\text{smooth} \gets \widehat{\xi}_{n|n} + C_n(\widehat{\xi}_{n+1|N}^\text{smooth} - \widehat{\xi}_{n+1|n}).
    \)
\EndFor

\Statex

\State \Return de-jittered signal (cf. \eqref{de_jitt_sig}): \(\widehat{x}_n \gets y_n - \widehat{\xi}_{n|N}^\text{smooth} y'_n \)

\end{algorithmic}
\end{algorithm}

\subsubsection{AR(1) process and white noise parameters estimation via conditional Maximum Likelihood} \label{subsec:MLEest}
To achieve optimal performance, the Kalman smoother requires exact knowledge of the three parameters \( \varphi \), \( \sigma_\epsilon \), and \( \sigma_w \). One way to obtain these parameter values is to estimate them during a preliminary calibration phase and then set them as constants. However, this solution would not account for potential time variations in the parameters. We propose in this section a conditional maximum likelihood estimation technique for \( \varphi \), \( \sigma_\epsilon \), and \( \sigma_w \) that can be repeated efficiently for each block of \(N\) samples. A numerical evaluation is deferred to Section~\ref{sec:num_sim}.

At pilot positions \( n \in I \) we can rearrange the second equation in \eqref{jitter_state-model} and define
\[
m_n \coloneqq \frac{y_n - x_n}{x'_n} = \xi_n + \frac{w_n}{x'_n},
\]
which is conditionally Gaussian given the realized value of \(x'_n\). In practice, \(x'_n\) is not available and we use the plug-in approximation \(x'_n \approx y'_n\); accordingly, all conditional statements below are understood with respect to the realized values of \(y'_n\).

Assume that pilots are uniformly spaced with \(K\) data samples between two pilots, i.e., consecutive pilots are separated by \(K{+}1\) time indices. Then the conditional covariance structure of \(\{m_n\}_{n\in I}\) can be written as
\begin{equation} \label{eq:mcov_struct} \begin{split}
    \gamma_m(n;d) & \coloneqq \mathbb{E}\!\left[m_{n+d(K+1)}\, m_n \,\middle|\, \x' \right] \\ &
    = \frac{\sigma_\epsilon^2}{1-\varphi^2}\,\varphi^{|d|(K+1)} \;+\; \frac{\sigma_w^2}{(x'_n)^2}\,\chi_{\{0\}}(d),  \end{split}
\end{equation}
where \(d\in\mathbb{Z}\) is the pilot lag (so that \(n+d(K+1)\in I\)).

If \( N \in \mathbb{N} \) samples are collected by the ADC, the covariance in \eqref{eq:mcov_struct} can be written as an \( M \times M \) symmetric real matrix with \(M=\lceil N /(K+1) \rceil\), that we decompose as \( \boldsymbol{\Sigma} + \mathbf{D} \), where \( \boldsymbol{\Sigma}_{ij} = \sigma_\epsilon^2 \varphi^{(K+1)|i-j|} / (1 - \varphi^2) \) has a Toeplitz Kac-Murdoch-Szeg\H{o} structure and \( \mathbf{D} = \text{diag}(\sigma_{w_{i_n}}^2) \). As already noted, \( \mathbf{m}  \) is conditionally Gaussian, and thus its conditional negative log-likelihood function \( \ell \) given the measurements \( \widetilde{\mathbf{m}} \) is
\begin{equation} \label{eq:negloglkhd}
\ell(\boldsymbol{\theta}; \widetilde{\mathbf{m}}) = \frac{M}{2} \log( 2 \pi) + \frac{1}{2} \log(\text{det}( \boldsymbol{\Sigma} + \mathbf{D})) + \frac{1}{2} \widetilde{\mathbf{m}}^t ( \boldsymbol{\Sigma} + \mathbf{D})^{-1} \widetilde{\mathbf{m}}
\end{equation}
with \( \boldsymbol{\theta} = (\sigma_\epsilon, \varphi, \sigma_w) \).
It can be shown that
\begin{equation} \label{eq:det}
\text{det}(\boldsymbol{\Sigma}) = \left( \frac{\sigma_\epsilon ^2}{1 - \varphi^2} \right)^M (1 - \varphi^{2(K+1)})^{M-1} \ne 0;
\end{equation}
additionally \( \boldsymbol{\Sigma}^{-1} \) is tridiagonal and has the closed-form expression
\begin{equation} \label{eq:inv}
\begin{split}
(\boldsymbol{\Sigma}^{-1})_{ij}
& = \frac{(1- \varphi^2)}{\sigma_\epsilon^2 (1- \varphi^{2(K+1)})} \\ &
\times
\begin{cases}
1 & \text{if } i = j = 1 \text{ or } i=j=M, \\
1 + \varphi^{2(K+1)} & \text{if } i=j, \, i \notin \{1,M\}, \\
-\varphi^{K+1} & \text{if } |i-j| = 1, \\
0 & \text{else.}
\end{cases}
\end{split}
\end{equation}
Therefore, if we write \( (\boldsymbol{\Sigma} + \mathbf{D})=\boldsymbol{\Sigma} (\mathbf{I} + \boldsymbol{\Sigma}^{-1} \mathbf{D}) \), the complexity of calculating \( \text{det}(\boldsymbol{\Sigma} + \mathbf{D}) \) is \( O(M) \) because \( (\mathbf{I} + \boldsymbol{\Sigma}^{-1} \mathbf{D}) \) is tridiagonal and \( \det (\boldsymbol{\Sigma}) \) is known from \eqref{eq:det}.

Similarly, we can use the Woodbury matrix identity to deal with the inversion of \( \boldsymbol{\Sigma} + \mathbf{D} \):
\[
(\boldsymbol{\Sigma} + \mathbf{D})^{-1}
= \boldsymbol{\Sigma}^{-1} - \boldsymbol{\Sigma}^{-1}(\mathbf{D}^{-1} + \boldsymbol{\Sigma}^{-1})^{-1} \boldsymbol{\Sigma}^{-1},
\]
from which we can deduce that the complexity of evaluating \( \widetilde{\mathbf{m}}^{H} (\boldsymbol{\Sigma} + \mathbf{D})^{-1} \widetilde{\mathbf{m}} \) is also \( O(M) \), since \((\mathbf{D}^{-1} + \boldsymbol{\Sigma}^{-1})\) is tridiagonal. Hence, since the evaluation of \eqref{eq:negloglkhd} can be performed very efficiently by exploiting the problem structure and using standard matrix theory tricks, we can, in turn, efficiently minimize \eqref{eq:negloglkhd} to obtain the maximum likelihood estimate of \( \boldsymbol{\theta} \) for use in the Kalman filter routine.

\begin{figure*}[!b]
  \centering
  % first row
  \subfloat[\( \sigma_\xi / T_s = 5 \cdot 10^{-3} \), \( \varphi = 0.999 \)]{%
    \includegraphics[width=0.49\textwidth]{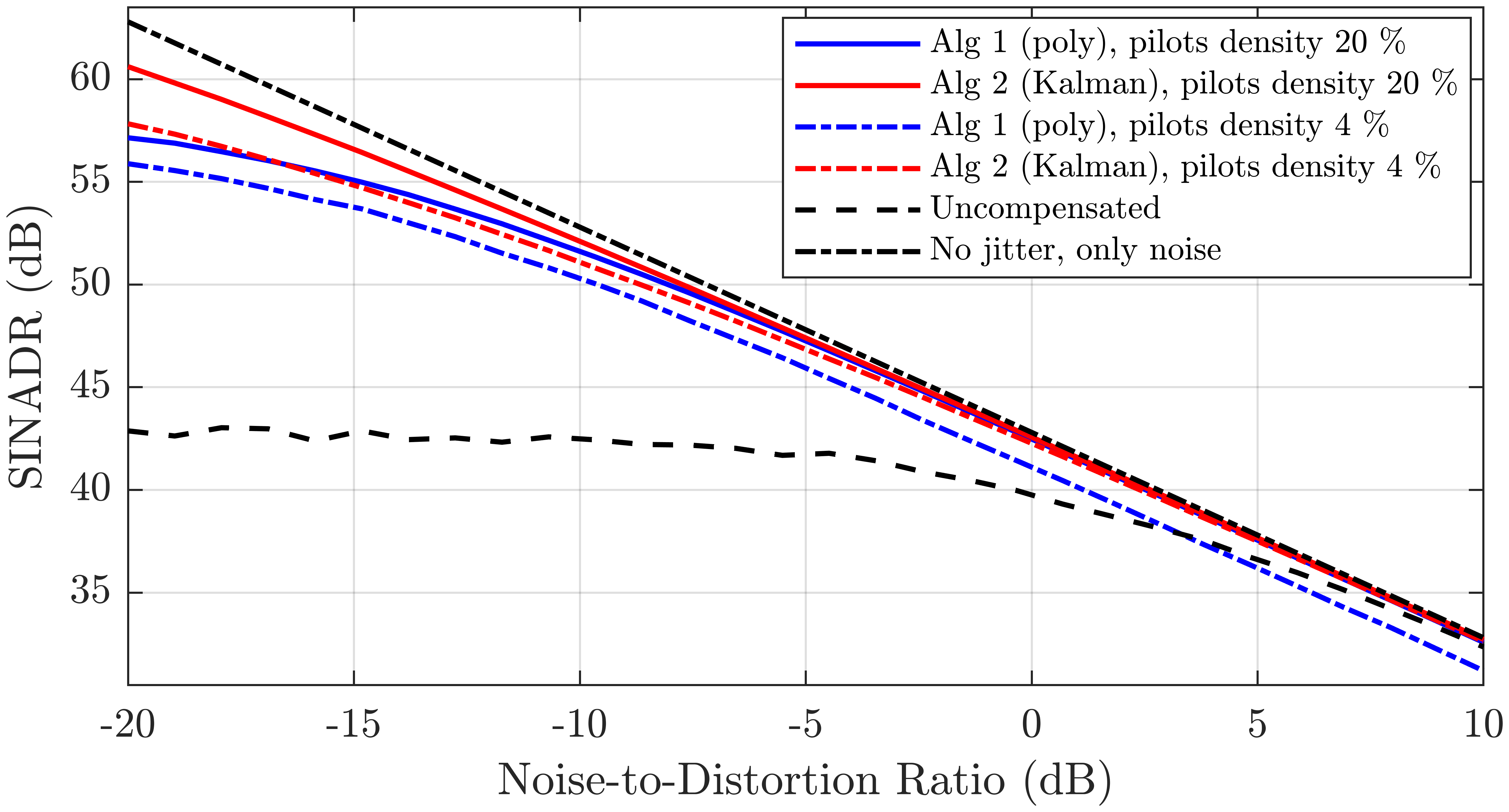}%
    \label{fig:0p5jitt}}%
  \hspace{0.01\textwidth}%
  \subfloat[\( \sigma_\xi / T_s = 1.5 \cdot 10^{-2} \), \( \varphi = 0.999 \)]{%
    \includegraphics[width=0.49\textwidth]{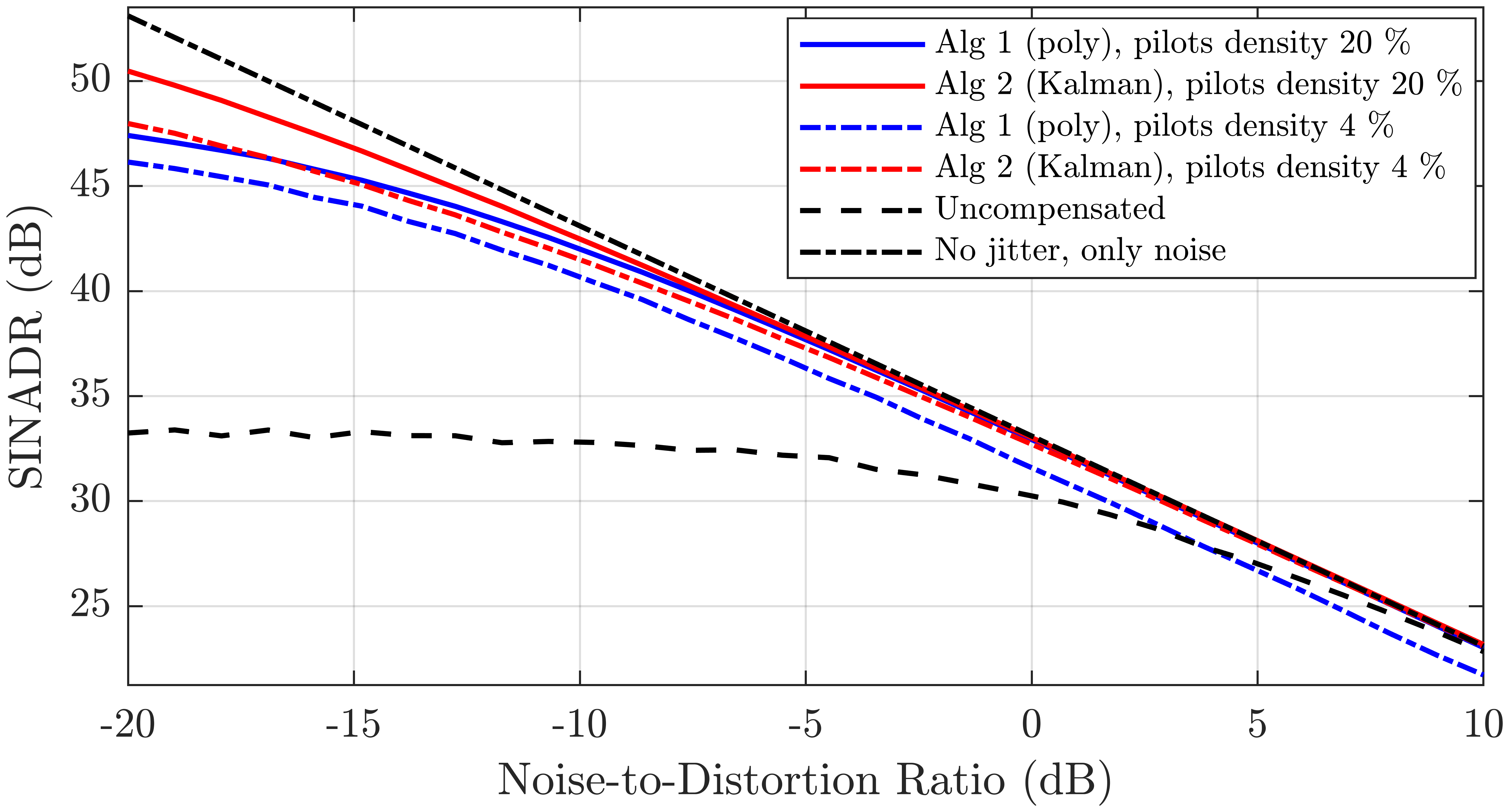}%
    \label{fig:1p5jitt}}
    \caption{Varying NDR analysis for two different jitter levels. The white noise power is progressively increased. The ``No jitter, only noise'' line shows \( \sigma_x ^2 / \sigma_w ^2 \) in dB.}
  \label{fig:wide_jitt_grid}
\end{figure*}

\subsection{Complexity analysis}
\label{sec:complanalysis}
The Kalman smoother processes each sample within the stream; a single iteration in both forward and backward passes has constant cost, thus the total cost for \(N\) samples is \( O(N) \), no matter what the pilot density \(p = (L(C-1) + 1)/N\) is.

In the polynomial technique, the matrix to be inverted in \eqref{weighted_MPinv} is \( (d+1) \times (d+1) \); the full cost of \eqref{weighted_MPinv}, including matrix multiplications, is thus \( O(C (d+1)^2 + (d+1)^3) \). Since this is repeated \(L\) times, the total cost for calculating the polynomial coefficients will be \( O(L(C (d+1)^2 + (d+1)^3)) \); the polynomial evaluations cost \( O(N(d+1)) \), bringing the total complexity to \( O(N(d+1) + L(C (d+1)^2 + (d+1)^3))  \). Rewriting the latter as function of the pilot density provides a complexity of \[ O\left(N\left((d+1) + p (d+1)^2 \frac{C + d+1}{C-1}\right) + (d+1)^2 \frac{C + d+1}{C-1} \right). \]

\section{Simulation results}

\label{sec:num_sim}

\subsection{Baseline evaluation}

Throughout most of this section, the test signal \(x\) is a bandlimited Gaussian process (cf. Section~\ref{sec:sys_model}) with a cutoff frequency of \(40\)~MHz, sampled at \(100\)~MS/s and normalized so that \(\sigma_x^2=1\).
The discrete-time waveform was generated by applying the desired brickwall filter to a sequence \(u\) of \(N=2^{18}\) independent real-valued Gaussian random variables.
The samples were taken at the uniform time instants \(nT_s\), and the jittered samples \(x(nT_s+\xi_n)\) were obtained by evaluating the underlying bandlimited continuous-time waveform via a truncated version of the Whittaker--Shannon interpolation formula,
\begin{equation}\label{formula:Whittaker_Shannon}
x(nT_s + \xi_n) = \sum_{m=-\infty}^{\infty} x(mT_s)\,
\sinc\!\left(\frac{nT_s+\xi_n-mT_s}{T_s}\right).
\end{equation}

We designated a subset of sample indices \(I \subset \{0,\dots,N-1\}\) as \emph{pilot positions} and assumed that the corresponding nominal (unjittered) samples \(\{x(nT_s)\}_{n\in I}\) were available as references; the effects of pulse shaping and matched filtering were therefore neglected in this section.
In a first set of simulations, we tested the relationship between background white noise and jitter, and their effects on the proposed algorithms. In Figure \ref{fig:wide_jitt_grid}, we fix the amount of jitter at two levels. We swept over \(30\) different white noise power levels (and ran \(5\) experiments each) so that the NDR ranges from \( -20  \) to \( 10 \) dB, and we assessed the overall SINADR (cf. formula \eqref{SINADR}). Two different pilot-sample densities are considered as well, while keeping fixed both the block length \( C = 500 \) and the polynomial degree \(d=4\) fixed for Algorithm \ref{alg:poly_alg}. The Kalman smoother has oracle knowledge of \( \sigma_w ^2 \), \( \sigma_\epsilon ^2 \) and \( \varphi = 0.999 \). As expected, the Kalman smoother displays superior performance throughout the entire NDR range; improvements after compensation up to \( 1-2 \) dB, after which the jitter disturbance is essentially drowned in noise and thus no longer distinguishable.

Another batch of simulations is shown in Figure \ref{fig:Pdensity_noise-15dB} where the NDR is kept constant, while the density of the (uniformly spaced) pilot samples is gradually increased from \( 1 \% \) to \( 20 \% \). 

\begin{figure}[H]
    \centering
    % First subfigure
    \begin{subfigure}[b]{0.49\textwidth}
        \centering
        \includegraphics[width=\textwidth]{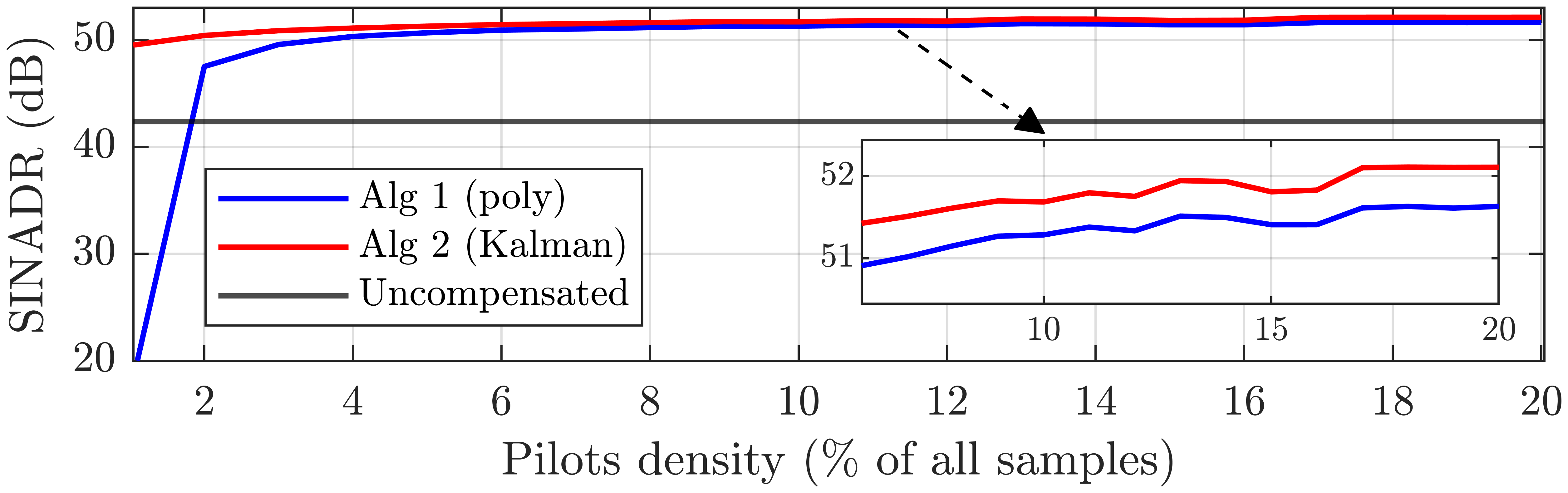}
    \end{subfigure}
    \begin{subfigure}[b]{0.49\textwidth}
        \centering
        \includegraphics[width=\textwidth]{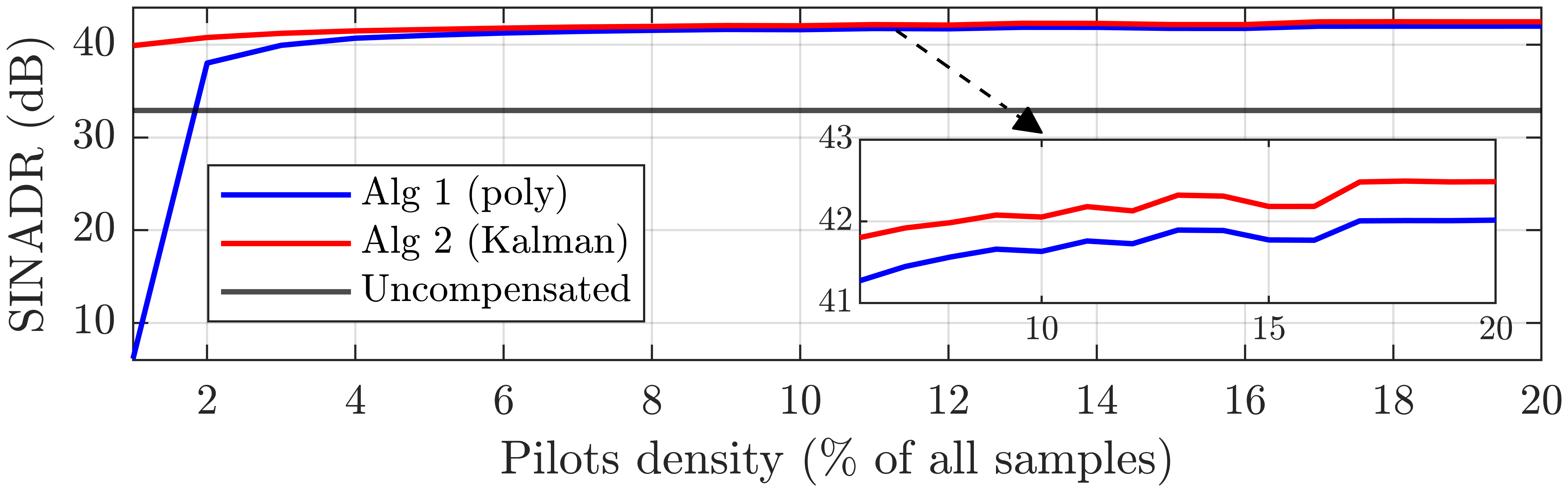}
    \end{subfigure}
    \caption{From top to bottom: \( \sigma_\xi / T_s = 5 \cdot 10^{-3}, \ 1.5 \cdot 10^{-2} \) respectively. NDR is fixed to \( = -10 \text{ dB}  \).}
    \label{fig:Pdensity_noise-15dB}
\end{figure}

Both techniques benefit from the higher pilot density, with a notable jump in performance for Algorithm \ref{alg:poly_alg} when going from \(1\) to \(2 \%\); at the same time, the Kalman smoother performs well even with pilot density as low as \(1 \%\).

In Figure \ref{fig:jitt_sweep} we show the algorithms' performance for increasing jitter levels at constant noise power; the latter is chosen so that the NDR is \(-10\) dB at the lowest jitter level. The compensated signals achieve \(6 - 15\) dB improvements in SINADR up to \( \approx 4 \% \) jitter, beyond which the polynomial method overtakes the Kalman smoother and continues to deliver consistent gains up to \(10 \% \) of jitter. 

\begin{figure}[h!]
        \centering
        \includegraphics[width=0.49\textwidth]{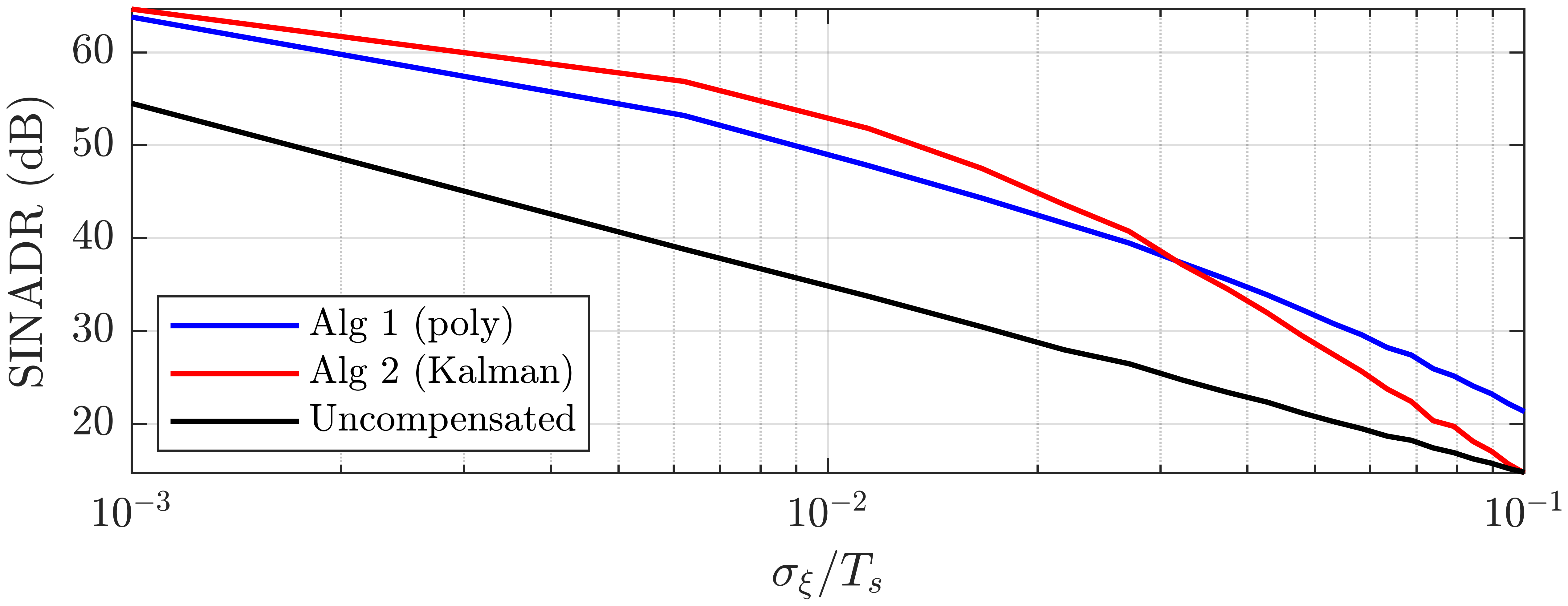}
    \caption{Fixed noise power (\(-10 \text{ dB}\) NDR at lowest jitter level, \( -42 \text{ dB}\) NDR at \( 4\% \)), increasing jitter.}
    \label{fig:jitt_sweep}
\end{figure}

This numerical observation highlights one limitation of a model-based approach, such as the Kalman filter. When the standard deviation of the jitter process is sufficiently large so that the first-order Taylor expansion becomes increasingly unreliable, the measurement model in \eqref{jitter_state-model} becomes invalid. Although higher-order Taylor terms may be lumped into \(w_n\), the latter would become colored and both signal- and jitter-dependent. Moreover errors due to model mismatch introduced in the Kalman recursion at certain time instants propagate, causing additional inaccuracies. In contrast, the polynomial interpolation technique, being block-based, does not propagate errors and is more robust to misspecified noise variance, especially if the polynomial degree is low.

We note that the parameters \(d\) and \(C\), together with the pilot density, interact in the polynomial technique. In these experiments, we did not attempt to identify their optimal combination; instead, we chose sub-optimal values empirically to avoid both over-fitting and under-fitting. The problem of optimally interpolating Wiener-type processes with polynomial functions has already been studied in, e.g., \cite{Blu_Unser}. Thus, a more rigorous investigation of this interplay is left for future work.

A final numerical simulation is shown in Figure \ref{fig:bw_sweep}, where we swept over different signal bandwidths to show the SINADR of the compensated signals, in a manner similar to Figure \ref{fig:wide_jitt_grid}. As the bandwidth increases, we kept either the NDR (left) or the SNR (right) constant.

\begin{figure}[h!]
        \centering
    % First subfigure
    \begin{subfigure}[b]{0.5\textwidth}
        \centering
        \includegraphics[width=\textwidth]{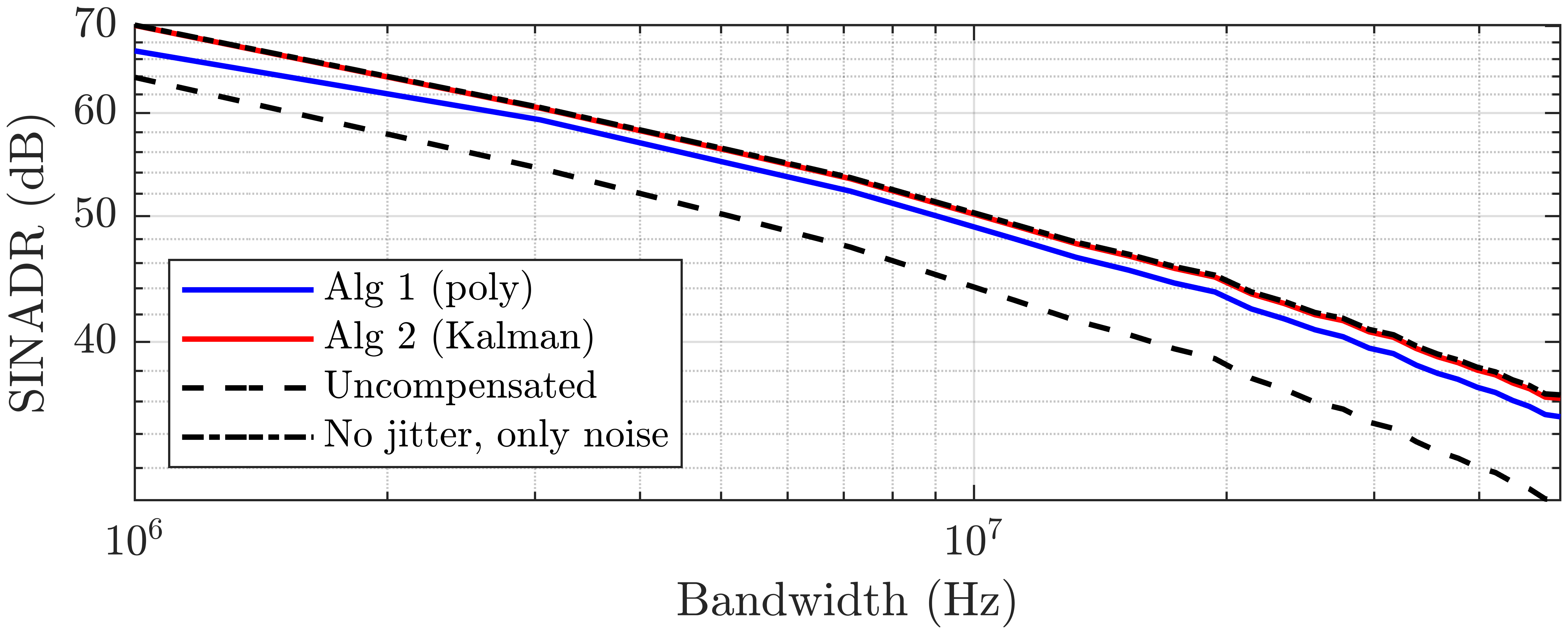}
    \end{subfigure}
    \begin{subfigure}[b]{0.5\textwidth}
        \centering
        \includegraphics[width=\textwidth]{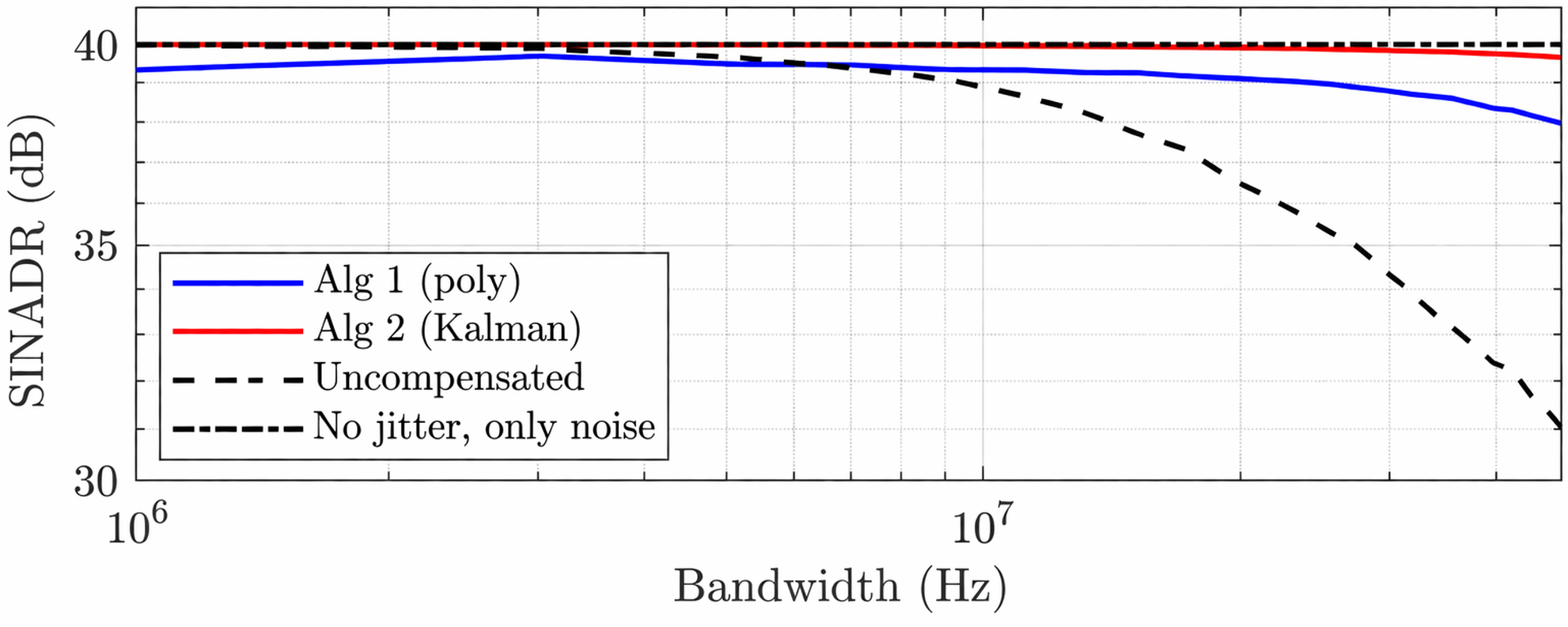}
    \end{subfigure}
    \caption{Algorithm performances vs signal bandwidth. Fixed NDR \(= -5 \) dB (top) and fixed SNR \( = 40 \) dB (bottom). \( \varphi = 0.999 \), \( \sigma_\xi / T_s = 1.5 \cdot 10^{-2} \).}
    \label{fig:bw_sweep}
\end{figure}

\subsection{Robustness of the Kalman smoother to parameter inaccuracies}
We numerically tested the parameters estimation technique described in Section \ref{subsec:MLEest}, and the results are summarized in Figure \ref{fig:ex_vs_est_params}. The plot shows a negligible performance loss when the parameters obtained by minimizing \eqref{eq:negloglkhd} are used in Algorithm \ref{alg:KalmanF}. Due to the non-convexity of \eqref{eq:negloglkhd} with respect to \( \boldsymbol{\theta}\), a careful optimization strategy is required. In our Monte Carlo simulations, we used the built-in MATLAB \texttt{fmincon} solver coupled with iterative randomization of the starting point with \(1000\) different initializations and some crude constraints on the parameter values. The test signal and the parameters values were kept fixed throughout the simulation, while the white-noise and jitter realizations were regenerated at each run.
\begin{figure}[h!]
  \centering
    \includegraphics[width=0.49\textwidth]{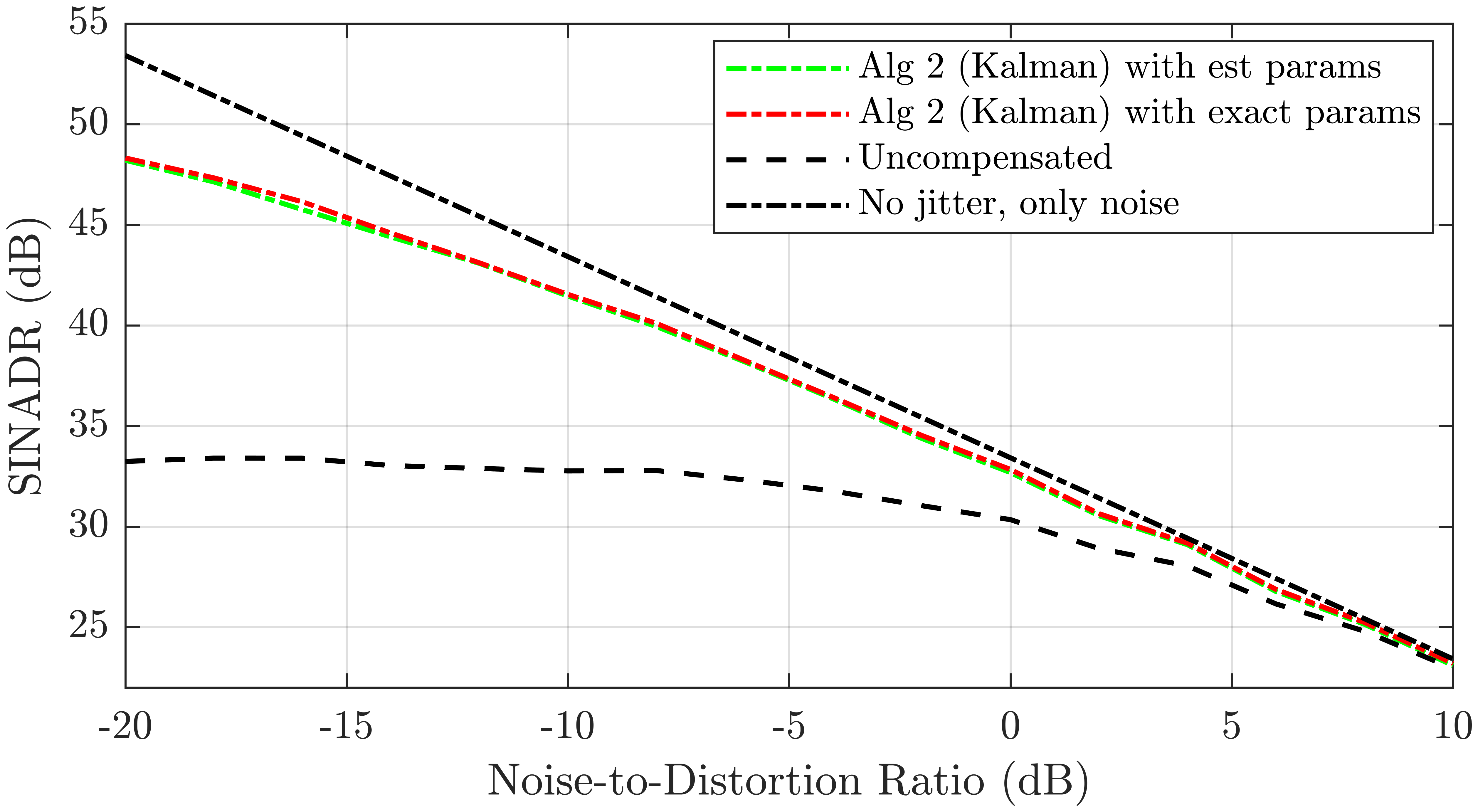}%
    \caption{Kalman smoother SINADR performance with noise and process parameters estimated from data. \(\varphi= 0.999\) and \( \sigma_\xi / T_s = 1.5 \cdot 10^{-2} \).}
    \label{fig:ex_vs_est_params}
\end{figure}

 \subsection{Radar use case example: suppression of jitter from nearly co-channel interferer}

The numerical examples displayed in the previous section consider jitter suppression as an in-band distortion; the power spectral densities of compensated and uncompensated signals are mostly similar because the term \( \xi x' \) has smaller variance than the signal itself if \eqref{small_jitter_hyp} holds. This can be rigorously proved by modifying the proof of Proposition \ref{yprim_approx_xprim}. The fact that the term \( \xi x' \) produces mainly in-band distortions is not entirely obvious, because its power spectral density is not compactly supported. However, its out-of-band contributions are negligible. This can be seen by studying the integral function \begin{equation} I(\nu) = \int_{-\overline{\omega}}^{\overline{\omega}} \frac{ \omega^2}{ 1 + \varphi^2 - 2 \varphi \cos(\nu - \omega) } \, d \omega,  \end{equation} which exhibits two peaks at \( |\nu| \approx \overline{\omega}  \) and decays rapidly once \( |\nu| > \overline{\omega} \).

To better illustrate the spectral impact of the proposed methods, we outline a different (and somewhat idealized) scenario motivated by radar applications. We assume a bistatic pair consisting of an illuminator of opportunity (IO) at known location and an antenna receiver with two channels, one for surveillance (SC) and one for reference (RC). The RC acquires high-fidelity copies of the IO waveform, while the SC primarily captures targets' echoes. However, the SC is also corrupted by direct-path interference (DPI) which is usually much stronger than echoes and can thus completely overshadow them.

\begin{figure}[h!]
  \centering
  \subfloat[Kalman\label{fig:blocker_scenarioKalman}]{
    \includegraphics[width=0.47\columnwidth]{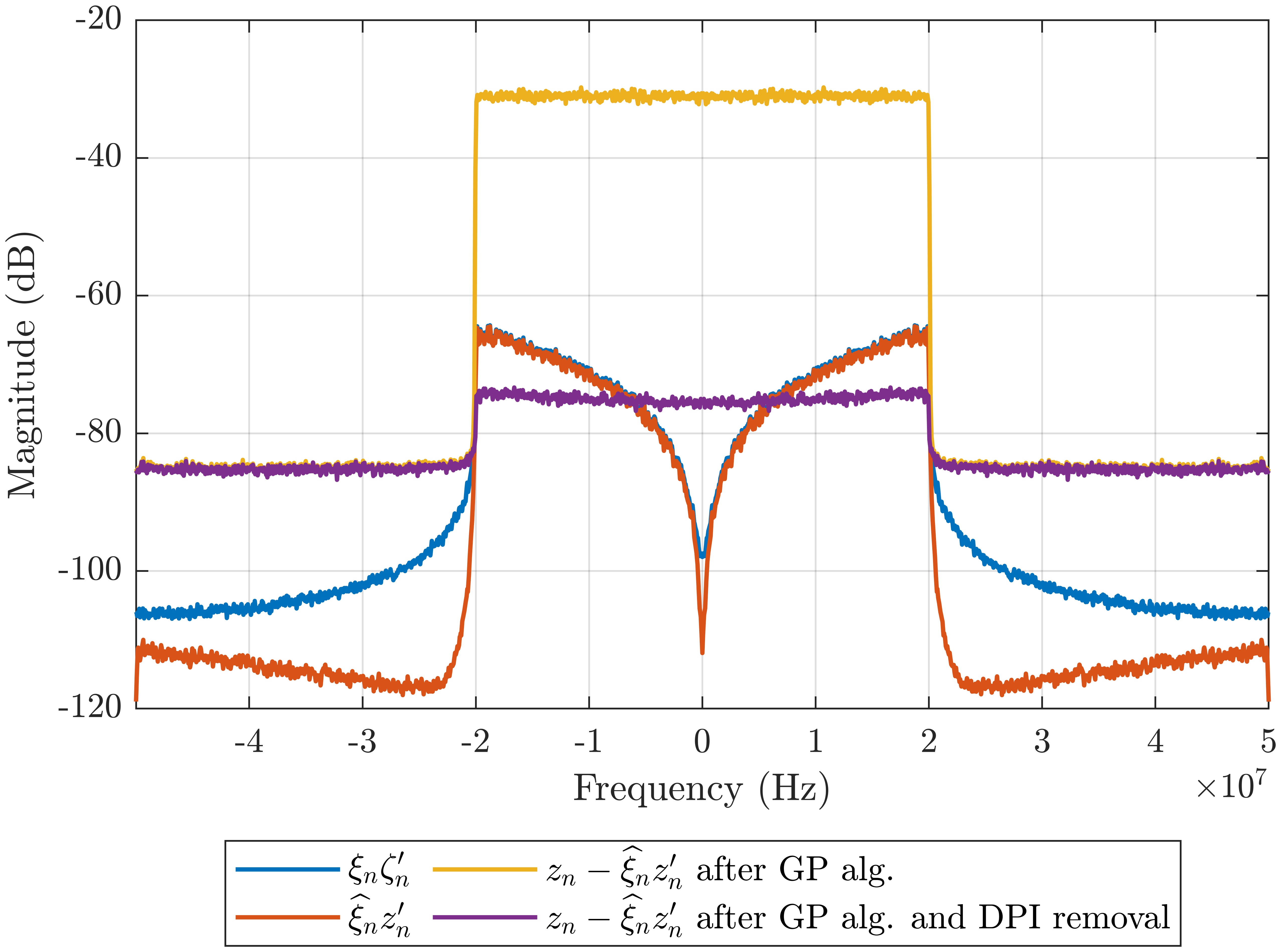}
  }
  \hfil
  \subfloat[Polynomial\label{fig:blocker_scenarioPoly}]{
    \includegraphics[width=0.47\columnwidth]{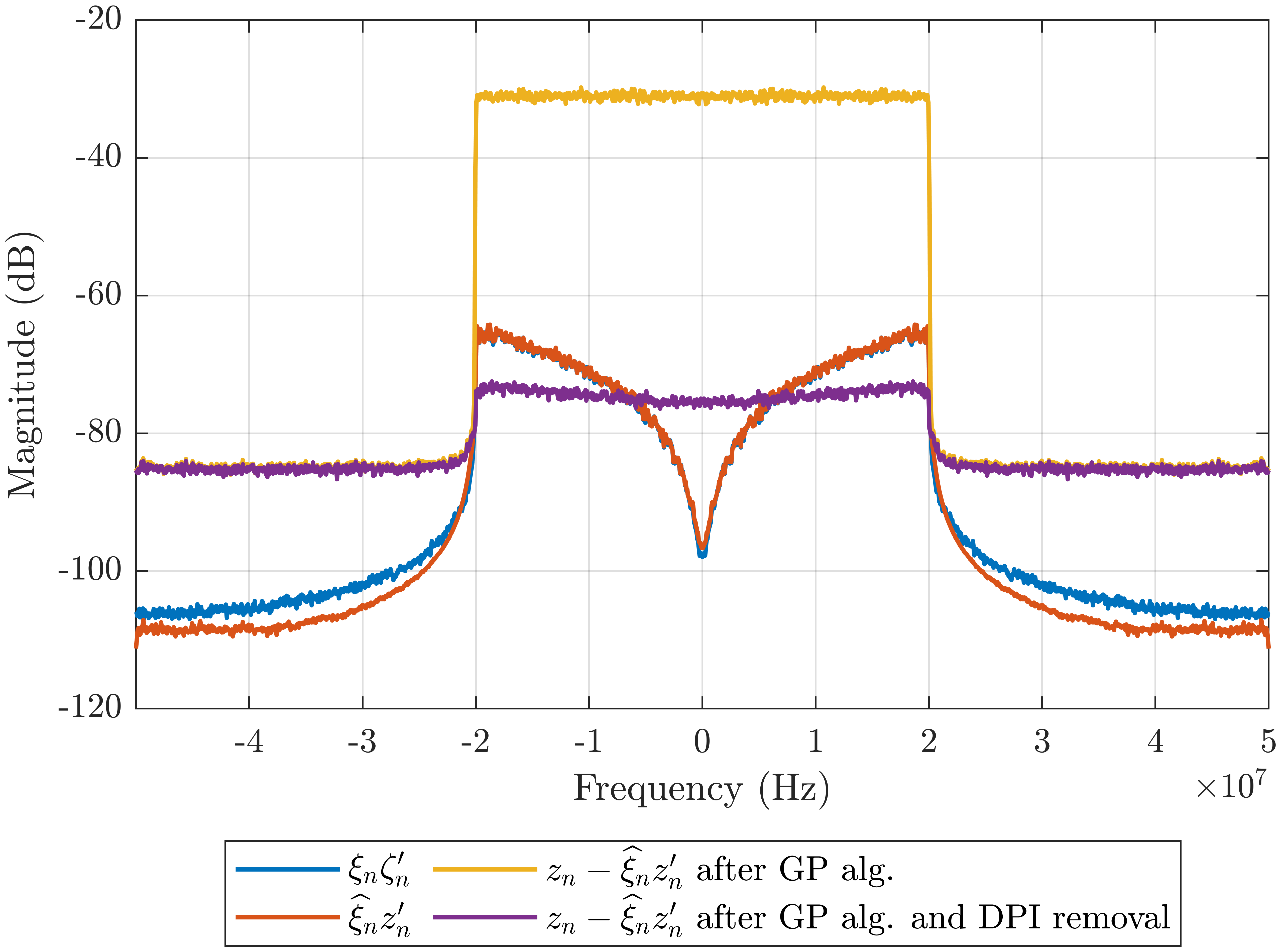}
  }
  \caption{Power spectral density illustrations at different compensation stages.}
  \label{fig:blocker_comp}
\end{figure}

The signal at the ADC output can be written as \begin{equation} \label{eq:radarEx} z_n = x_n + \xi_n x_n ' + \zeta_n + \xi_n \zeta_n ' + w_n. \end{equation} Here, \( \zeta_n \) is the contribution from the DPI, while \( x_n \) is the wanted signal. The term \( \zeta_n + x_n  \) is assumed to be known if \( n \in I \). By assumption \( \sigma_\zeta ^2 \gg \sigma^2 _x   \) and \( \sigma_{\xi \zeta'} ^2 > \sigma^2 _x\). This means that even the interferer-induced jitter term dominates the signal of interest.

We consider a situation in which the term \( \zeta_n \) can be canceled with negligible impact on the other terms (at least if the target is moving and there is a nontrivial Doppler component in \(x_n\)), for example by projecting it away using subspace methods in combination with the RC signal knowledge, see, e.g. \cite{Coloneetal}. After dejittering \(z_n\) using the methods presented in the previous section and optimally interpolating the ``holes'' left after removing the pilot samples using the Gerchberg-Papoulis algorithm, we then project away \( \zeta_n \), and what is left is the desired echo signal. Figure \ref{fig:blocker_comp} shows examples of power spectral densities of the different compensation stages described in this section.

\subsection{Relation to existing techniques}
In a significant portion of the existing literature, clock jitter is modeled as Gaussian and spectrally white, and the mitigation schemes proposed there are designed accordingly. In that context, the notion of ``tracking'' is vague. Adapting those techniques to the correlated-jitter scenario is non-obvious, thus extensive work is needed and deemed beyond the scope of the present work. The work closest in spirit to ours is \cite{Towfic_Ting_Sayed}, where clock jitter is modeled with a \(1/f^2\) spectral decay and the jitter process is tracked by placing a sideband pilot tone that is first downconverted to DC and then lowpass-filtered. The resulting phase oscillations are extracted by minimizing an expected value, which results in an infinite-impulse-response (IIR)-type recursion. We therefore benchmark our routines with the technique in \cite{Towfic_Ting_Sayed}.

We nonetheless stress that our algorithms and \cite{Towfic_Ting_Sayed} remain structurally different: our proposed methods exploit known pilot samples intertwined with the payload, whereas \cite{Towfic_Ting_Sayed} separates pilots and payload in frequency via a pilot tone. It is not entirely clear how to construct a scenario in which all relevant quantities (bandwidth usage, pilot-symbol density, pulse-shaping filters, etc.) can be consistently mapped across the three techniques; nevertheless, it is still possible to draw qualitative conclusions, which we do in this section.

We generated \(N_{\text{sym}} = 43691 \) Gaussian data symbols \(s_i\) and we pulse-shaped them into a real-valued \(40 \text{ MHz}\)-wide (from \(-20\) to \(20\) MHz) waveform using a windowed sinc interpolation, and the resulting analog waveform \( s (t)\) was sampled at \( 120 \text{ Msps} \) at the ADC side. The symbols at uniformly spaced indices \(I \subseteq \{1, \dots, N_{\text{sym}}\} \) were designated as pilots; those at \(I_d\) as data. Since the oversampling factor is an integer, there is no intersymbol interference (ISI) at the symbol-synchronous sampling instants, and therefore knowledge of sample values at those instants can be inferred from pilot symbol values.

The analog waveform \(s(t)\) serves as a common test signal for all methods so that there are no differences in digital signal processing operations that could lead to different values of the baseline \(\text{EVM}_{\text{uncomp}} \) \eqref{eq:EVM}. For \cite{Towfic_Ting_Sayed}, the pilot tone is injected outside the payload band (hence using a larger effective bandwidth) and it is removed prior to symbol recovery, so that the baseline \(\text{EVM}_{\text{uncomp}}\) is evaluated on the payload component in \(I_d\) only.

The optimal forgetting factor in the infinite impulse response-type recursion \cite{Towfic_Ting_Sayed} was computed using oracle knowledge of the jitter process realization, not feasible in practice, while the parameters \(\varphi \), \( \sigma_\epsilon \) and \( \sigma_w \) were estimated using the Maximum Likelihood routine described in section \ref{subsec:MLEest}. In the polynomial method, we have not optimized the degree and the batch lengths.

After inserting jitter and white noise, we performed jitter estimation and compensation. We then downsampled and applied the matched filter, along with the optimal rescaling factor, to obtain estimates \( \{ \widehat{s}_i \}_{i \in I_d} \) of the data symbol values. The performances are here assessed via the \(\Delta\) Error Vector Magnitude (EVM) defined as follows:
\begin{equation} \label{eq:DeltaEVM}
\Delta \text{EVM}_{\text{dB}} \coloneqq 20 \log_{10} \left( \frac{\text{EVM}_{\text{uncomp}}}{\text{EVM}_{\text{comp}}} \right)
\end{equation}
and
\begin{equation} \label{eq:EVM}
\text{EVM}_{\text{comp}} = \sqrt{ \frac{ \frac{1}{|I_d|} \sum_{i \in I_d} |\widehat{s}_i - s_i|^2  }{\frac{1}{|I_d|} \sum_{i \in I_d} | s_i|^2} }
\end{equation}
and equivalently for \(\text{EVM}_{\text{uncomp}} \), replacing the \( \widehat{s}_i \) with their uncompensated counterparts.

\begin{table*}[!t]
\caption{Comparative summary of the two proposed methods}
\label{tab:threecol_bottom}
\centering
\renewcommand{\arraystretch}{1.15}
\setlength{\tabcolsep}{8pt}

 % <-- everything inside (text + rules) becomes red
\begin{tabularx}{\textwidth}{>{\raggedright\arraybackslash}X
                                >{\centering\arraybackslash}X
                                >{\centering\arraybackslash}X}
\toprule
\textbf{} & \textbf{Polynomial interpolation \ref{alg:poly_alg}} & \textbf{Kalman smoother \ref{alg:KalmanF}}\\
\midrule
High-level complexity (cf. Section \ref{sec:complanalysis})  & \(O\left(N\left((d+1) + p (d+1)^2 \frac{C + d+1}{C-1}\right) + (d+1)^2 \frac{C + d+1}{C-1} \right) \) & \(O(N)\)\\
Robustness to low pilot density (cf. Figure \ref{fig:Pdensity_noise-15dB}) & \(3-4 \%\) density required for optimal performance & As low as \(1 \%\) density is enough for optimal performance \\
Robustness to model mismatch (cf. Figure \ref{fig:jitt_sweep})  & Less sensitive to measurement model mismatch & More sensitive to measurement model mismatch \\
Robustness to parameter inaccuracies (cf. Section \ref{subsec:MLEest}) & Model agnostic & Parameters can be accurately estimated \\
\bottomrule
\end{tabularx}
 % end color group

\end{table*}

The outcome of this qualitative analysis is shown in Figure \ref{fig:comp_literature}. We plot the \( \Delta \text{EVM}_{\text{dB}} \) on the \(y\) axis whereas on the \(x\) axis we report two overlaid quantities, \( p = |I| / N_{\text{sym}} \) (pilot symbol density) for our algorithms and \% \(p_t\) of the unit total power budget allocated to the pilot tone for \cite{Towfic_Ting_Sayed}. By construction \( \text{EVM}_{\text{uncomp}} \) is constant across methods and across all density/power levels. Note that while \( p = 0 \Longleftrightarrow p_t = 0 \) (no pilots), the relationship between the two is not necessarily linear.

\begin{figure}[h!]
  \centering
    \includegraphics[width=0.49\textwidth]{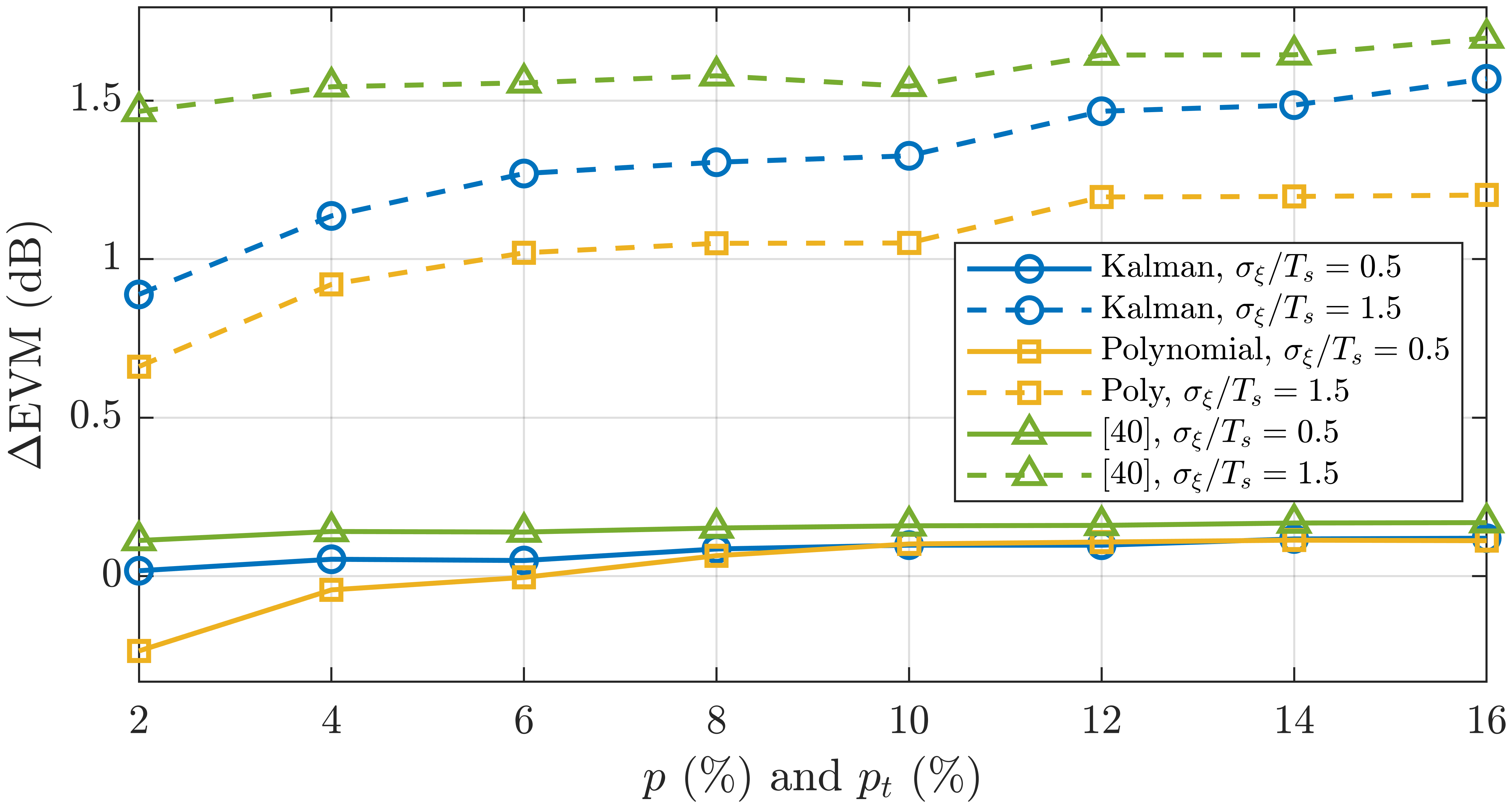}%
    \caption{EVM performances. On the \(x\) axis: pilot-symbol density for Algorithms~\ref{alg:poly_alg} and~\ref{alg:KalmanF}, and pilot-tone fractional power for \cite{Towfic_Ting_Sayed}. \(\varphi= 0.999\).}
    \label{fig:comp_literature}
\end{figure}

Algorithms~\ref{alg:poly_alg} and~\ref{alg:KalmanF} display a qualitative behavior consistent with that seen in Fig.~\ref{fig:Pdensity_noise-15dB}: at low pilot density, Algorithm~\ref{alg:poly_alg} struggles to deliver a positive gain, and both curves somewhat flatten beyond a certain threshold. The method in \cite{Towfic_Ting_Sayed} performs better overall; however, the three methods are not directly comparable. Moreover, \cite{Towfic_Ting_Sayed} is favored by its use of additional bandwidth and oracle (genie) knowledge of the jitter process realization; at higher power levels \cite{Towfic_Ting_Sayed} is expected to display a performance loss due to low SNR \cite{GerosaAnttilaEriksson}.

\section{Conclusions and future research}
In this paper, we propose two algorithms to address stochastic distortions in baseband signals caused by local-oscillator timing jitter in analog-to-digital converters. We modeled the jitter as a first-order autoregressive process. The Kalman smoother approach shows robust and nearly optimal performance across a wide range of scenarios (cf. Figure \ref{fig:wide_jitt_grid}), and both methods display significant improvements in SINADR even with low pilot-sample densities (Figure \ref{fig:Pdensity_noise-15dB}). We also assess their impact when pushing the linear approximation \eqref{jitter_TaylorExp} to its limits, and we find that for very large jitter the polynomial compensation algorithm could be more suitable, as it is less sensitive to outliers and model mismatches (Figure \ref{fig:jitt_sweep}). Future research will focus on incorporating additional distortions (e.g. phase noise when the signal is upconverted to radio frequency, strong nearby blockers already contaminated with Digital-to-Analog Converter (DAC) jitter, power-amplifier nonlinearities, harmonic distortions, non-zero ISI) in the compensation routines.

\section{Acknowledgment}
We thank the anonymous reviewers and the Associate Editor for the many insightful and helpful comments and suggestions.

We acknowledge the use of ChatGPT (OpenAI) for grammar and clarity suggestions and as a sounding board to improve exposition. The authors are solely responsible for all technical content and conclusions.

\begin{appendix}
\begin{proof}[Proof of Proposition \ref{sec_order_varEst}]
This follows from basic properties of random variables. Indeed if \( \text{var} (x' _n) = \sigma_{x'} ^2  \), it is well-known \cite{Papoulis} that \( \text{var} (x ''_n) = 12 \pi^2 \sigma_{x'} ^2 W^2 / 5   \), under the assumption that \( \sigma_x ^2 = 1  \) and \( \mathcal{S}_x (f) = 1 / (2W) \chi_{ \{|f| \le W \}} (f)  \). 

The two processes \( \xi_n ^2  \) and \( x'' _n \) originate from unrelated sources and can thus be considered independent. Therefore we have \[ \begin{split} \text{var} (\xi_n ^2 x ''_n) & =  \text{var} (\xi_n ^2) \text{var}( x ''_n)  + \text{var} (\xi_n ^2) \underbrace{\mathbb{E}[x'' _n]^2}_{=0}   + \text{var} (x'' _n) \mathbb{E}[\xi_n ^2]^2 \\ & = 36 \pi^2 \sigma_\xi ^4 \sigma_{x'} ^2 W^2 / 5;  \end{split}  \]
similarly 
\[ \begin{split} \text{var} (\xi_n  x '_n) & =  \text{var} (\xi_n ) \text{var}( x '_n)  + \text{var} (\xi_n ) \underbrace{\mathbb{E}[x' _n]^2}_{=0}   + \text{var} (x' _n) \underbrace{\mathbb{E}[\xi_n ]^2}_{=0} \\ & =  \sigma_\xi ^2 \sigma_{x'} ^2.  \end{split}  \]
The conclusion follows if \(36 \pi^2 \sigma_\xi ^2 W^2 / 5 \ll 1  \), which is true under the ``small jitter'' hypothesis \eqref{small_jitter_hyp} and the fact that \( WT_s < 1 / 2 \), as a consequence of the Shannon-Nyquist sampling theorem. 
\end{proof}

  \begin{proof}[Proof of Proposition \ref{yprim_approx_xprim}]
We define \( \mathbf{u} \coloneqq \boldsymbol{\xi} \odot \x ' \) and we assume without loss of generality that \( \sigma_x^2 = 1 \). Moreover, to highlight the dependence on the factor \( W T_s \), we use here normalized frequencies \( \omega \in [- \pi, \pi] \). The spectral density of the product of two WSS discrete-time processes can be written as convolution \cite{Oppenheim_Schafer}\begin{equation} \label{first_conv} \mathcal{S}_{u} (e^{i\omega}) = \frac{1}{2 \pi} \int_{-\pi}^\pi \mathcal{S}_{\xi}( e^{i\nu}) \mathcal{S}_{x'}(e^{i( \omega - \nu)}) \, d \nu;  \end{equation} moreover a discrete autoregressive process of order \(1\) has power spectral density \cite{Kay} \[\mathcal{S}_{\xi}(e^{i\nu}) = \frac{\sigma_\epsilon ^2}{ 1 - 2 \varphi \cos(\nu) + \varphi^2}, \ \nu \in [-\pi, \pi]. \]  The power spectral density of our continuous-time signal model is \( \mathcal{S}_x (f) = 1/(2 W) \chi_{\{ |f| \le W \} } \) which after sampling becomes \( \mathcal{S}_x (e^{i\omega}) = 1/(2 W T_s) \chi_{\{ |(\omega)_{2 \pi}| \le 2 \pi W T_s \} } \), where the operator \( ( \cdot )_{2 \pi}\) denotes wrapping into \( [-\pi, \pi] \). In addition, the continuous-time frequency response of an ideal differentiator is \( H_D( \Omega) = i \Omega \) (cf. example 4.5 in \cite{Oppenheim_Schafer}), which becomes \( H_D(e^{i \omega}) = i \omega / 
T_s  \) in our setting. Thus \[\mathcal{S}_{x'} (e^{i\omega}) = |H_D (e^{i \omega})|^2 \mathcal{S}_{x} (e^{i\omega}) = \frac{((\omega)_{2 \pi}) ^2}{2 W T_s ^3} \chi_{\{ |(\omega)_{2 \pi}| \le 2 \pi W T_s \}}  \] and by applying once more the Wiener-Khinchin theorem we can conclude that
\begin{equation} \begin{split}
\mathcal{R}_{D \mathbf{u}} [0] & = \text{var}((D \mathbf{u})_n) \\ &  = \frac{1}{2 \pi} \int_{-\pi}^\pi |H_D(e^{i \omega})|^2 \mathcal{S}_u (  e^{i \omega})\, d \omega \\ & = B \int_{-\pi}^{\pi} \int_{-\pi}^{\pi} \frac{ \omega^2 ((\omega - \nu)_{2 \pi})^2}{1 - 2 \varphi \cos(\nu) + \varphi^2} \chi_{\{|(\omega - \nu)_{2 \pi}| \le 2 \pi W T_s\} } \, d \omega \, d \nu
\end{split}
\end{equation}
with \(B = \sigma_\epsilon ^2 /(8 \pi^2 W T_s ^5)\) and \eqref{var_derEst} follows by recalling that \( \sigma_{x'}^2 = 4 \pi^2 W^2 / 3  \) and \( \sigma_{\xi}^2 = \sigma_\epsilon ^2 / (1 - \varphi ^2) \).

\end{proof}

\end{appendix}
\end{document}